\documentclass[opre,nonblindrev]{informs3rad} 

\usepackage{endnotes}
\let\footnote=\endnote

\usepackage{accents}
\usepackage{xcolor}
\definecolor{cornellred}{rgb}{0.7, 0.11, 0.11}
\usepackage{hyperref}
\hypersetup{
    colorlinks = true,
    linkcolor=cornellred,
    citecolor=blue,
    linkbordercolor = {white}
}

\usepackage{natbib}
 \bibpunct[, ]{(}{)}{,}{a}{}{,}%
 \def\bibfont{\small}%
 \def\bibsep{\smallskipamount}%
\TheoremsNumberedThrough     
\ECRepeatTheorems

\EquationsNumberedThrough    

\MANUSCRIPTNO{OPRE-2018-03-077} 

\usepackage[ruled,vlined,linesnumbered]{algorithm2e}
\SetNoFillComment

\newenvironment{myprocedure}[1][htb]
{  
\begin{algorithm}[#1]%
}{\end{algorithm}}

\newcommand{\ubar}[1]{\underaccent{\bar}{#1}}

\begin{document}


\RUNAUTHOR{Niazadeh et al.}

\RUNTITLE{Fast Core Pricing for Rich Advertising Auctions}

\TITLE{Fast Core Pricing for Rich Advertising Auctions}

\ARTICLEAUTHORS{%

\AUTHOR{Rad Niazadeh}
\AFF{Chicago Booth School of Business, University of Chicago, \EMAIL{rad.niazadeh@chicagobooth.edu}}

\AUTHOR{Jason Hartline}
\AFF{Computer Science Department, Northwestern University, \EMAIL{hartline@eecs.northwestern.edu}}

\AUTHOR{Nicole Immorlica}
\AFF{Microsoft Research New England, \EMAIL{nicimm@microsoft.com}}

\AUTHOR{Mohammad Reza Khani}
\AFF{Amazon, \EMAIL{khani87@gmail.com}}

\AUTHOR{Brendan Lucier}
\AFF{Microsoft Research New England, \EMAIL{brlucier@microsoft.com}}
}

\ABSTRACT{
Standard ad auction formats do not immediately extend to settings where multiple size configurations and layouts are available to advertisers. In these settings, the sale of web advertising space increasingly resembles a combinatorial auction with complementarities, where truthful  auctions such as the Vickrey-Clarke-Groves (VCG) can yield unacceptably low revenue.  We therefore study core selecting auctions, which boost revenue by setting payments so that no group of agents, including the auctioneer, can jointly improve their utilities by switching to a different outcome. Our main result is a combinatorial algorithm that finds an approximate bidder optimal core point with almost linear number of calls to the welfare maximization oracle. Our algorithm is faster than previously-proposed heuristics in the literature and  has theoretical guarantees. We conclude that core pricing is implementable even for very time sensitive practical use cases such as realtime auctions for online advertising and can yield more revenue. We justify this
claim experimentally using the Microsoft Bing Ad Auction data, through which we show our core pricing algorithm generates almost 26\% more revenue than VCG on average, about 9\% more revenue than other core pricing rules known in the literature, and almost matches the revenue of the standard Generalized Second Price (GSP) auction.

}


\KEYWORDS{Sponsored search auctions, Core pricing, VCG auction, GSP auction, Sale of ad space, Combinatorial auction} 
\maketitle
\DoubleSpacedXI 

\section{Introduction}

Auctions with combinatorial preferences are prevalent in practice.
One prominent example is the \emph{sale of online advertising space} in search engines, also known as the \emph{rich advertising auction}~\citep{cavallo2017sponsored}, where advertisers with ads of various shapes, configurations, and decoration options have combinatorial and potentially complementary preferences over the space of ads. A common scenario is when each advertiser has several possible ads with varying sizes that take up a different number of lines. The
platforms selling the ad space (with limited number of lines) uses auctions to determine which of these ``rich ads"
will appear where, and at what price.


Standard solutions for a simple separable position environment -- which is  the common format for the sponsored search, e.g., see ~\citealp{edelman2007internet} -- cannot easily be extended to the sale of ad space setting. Auctions such as the Generalized Second Price (GSP) auction are typically ill-defined in combinatorial environments, as their payment rule is tied to the setting where the allocation problem is only ranking the winner ads. Such an auction can only be extended in ad-hoc ways to a general combinatorial auction --e.g., charging each advertiser the bid of the next advertiser in the order regardless of ad sizes-- which lack economic grounding.


We instead consider combinatorial auctions for the sale of ad space problem.  Among possible options, two are of particular note: (i) Vickrey-Clark-Groves (VCG) auction \citep{vickrey1961counterspeculation, clarke1971multipart,groves1973incentives} and (ii) core-selecting auctions~\citep{ausubel2006lovely,milgrom2007package}.  
These auctions are direct revelation mechanisms and select the optimal welfare allocation with respect to the reported bids, i.e., the one with the maximum total declared value for the bidders.  Optimizing welfare itself is a computationally hard problem in many cases of interest (e.g., see \citealp{nisan2007computationally,sandholm2002algorithm}).  However, even in settings where this difficulty can be resolved satisfactorily (e.g., via heuristics that have exponential runtime in the worst case but tend to solve practical instances quickly), we still need to address the problem of computing the payments.

In VCG autions, the payment of a bidder is the externality that
he imposes on other bidders by consuming the resources allocated
to him.  While rarely used in practice \citep{ausubel2006lovely}, this
auction is a useful point of comparison as it has several important
theoretical properties.  First, among welfare-maximizing auctions, it has the unique payment rule (up to additive offsets) that
incentivizes truthtelling as a dominant strategy.  Second, given access to an oracle which computes the optimal
welfare allocation, it is computationally efficient: it only takes
$n+1$ oracle calls to compute the allocation and payments, where $n$ is the number of bidders in the auction.

Unfortunately, in the presence of complementarities, the revenue generated by VCG can be quite low compared to the bidders' values, and the resulting outcome can seem unfair.  Consider, for example, a setting with two items, $A$ and $B$, and three bidders, $1$, $2$, and $3$.  Suppose bidder $1$ only wants item $A$ and has a value of $\$100$ for it.  Similarly, bidder $2$ only wants item $B$ and also has a value of $\$100$ for it.  Bidder $3$ has complementary preferences.  He only wants both $A$ and $B$ and has a value of $\$101$ for this bundle.  In this setting, the VCG auction gives item $A$ to bidder $1$ and item $B$ to bidder $2$ and charges each of them a price of $\$1$ for a total revenue of $\$2$.  This revenue is both low compared to the values, and also seemingly unfair from the point of bidder 3 who would be willing to pay quite a bit more than the winners.  

Core-selecting auctions attempt to address both the revenue and
fairness issues of VCG by achieving outcomes that are immune to renegotiations among coalitions, while retaining welfare optimality with respect to the reported bids.  Here, payments are set such that no group
of bidders, including the auctioneer, can simultaneously improve
outcomes (with respect to reported bids) by switching to different allocations and payments. In the
above example, if the bidders bid their true values, then any set of
payments such that bidders $1$ and $2$ jointly pay at least $\$101$
while each paying at most $\$100$ could be the outcome of a
core-selecting auction. Notably, given reported bids, all possible core payments form a polytope in $\mathbb{R}^n$ -- which we refer to as the \emph{core polytope}. This polytope has economic grounding in cooperative game theory~\citep{osborne1994course}.

Core-selecting auctions are not truthful in the sense of dominant strategies, but there are core selecting auctions that minimally sacrifice the incentives. In particular,  \emph{Pareto optimal} -- also known as \emph{bidder optimal} -- core points are of interest in this paper. In these auctions,  it is not possible to reduce any bidder's payment and still remain in the core.  In the above example, having each of bidders $1$ and $2$ pay $\$100$ is a core outcome, but is not bidder-optimal; having bidder $1$ pay $\$50$ and bidder $2$ pay $\$51$ \emph{is} bidder-optimal. As demonstrated in \citealp{day2007fair,day2008core}, picking a bidder optimal core point minimizes the maximum of agents' utilities from any deviations from truthful reporting, subject to selecting a core outcome. Moreover, these auctions have natural (full-information) equilibria that generate welfare-optimal outcomes; see also \citealp{day2012quadratic} for a discussion on desirable practical properties of bidder-optimal core points and their variations. As a result, they have been used in practice to sell wireless spectra~\citep{cramton2013spectrum} and more recently are proposed for use in selling online advertising~\citep{goel2015core}.

To be truly practical for large frequent auctions like sponsored search ad auction, an auction must be (i)~{highly computationally
  efficient}, (ii)~{simple enough for implementation purposes}, and
(iii)~{extendable to
different economic objectives}. As core-selecting auctions output
optimal-welfare allocations, the computation associated with finding
such allocations is unavoidable. However, in practice and in
particular for the application of sale of ad-space, the
problem is typically structured so that it is possible to compute
welfare-optimal (or near-optimal) allocations quickly for realistic
preferences. In fact, as we elaborate more later, bidding languages commonly used in practice for sponsored ad
space express preferences for which one can compute (near)
optimal outcomes. The main issue we tackle in this paper is the
additional difficulty of computing bidder-optimal core payments, given
a a satisfactory solution to the welfare optimization problem. 

To make this separation clear, we assume oracle access to a slightly more general version of the welfare optimization algorithm.
This oracle can handle  \emph{truncations} of reported bids,
in which each bidder's bid for each bundle is shifted by the same additive offset (for more details see Definition~\ref{def:win-det}). Notice that while this oracle might be as computationally tractable as the vanilla welfare-optimization oracle -- which is indeed the case for ad auctions with simple bidding languages; see Section~\ref{sec:experiment}--  they can potentially impose further computational challenges.  Nevertheless, this form of oracle is common in the core auction literature~\citep{ausubel2002ascending}. Prior work on core computation assume access
to such an oracle and compute core outcomes using either
heuristics
or computationally intensive convex programming/linear programming methods~\citep{day2008core,day2007fair,erdil2010new,bunz2015faster}.



\subsection{Main results} 
We consider both the theoretical problem of designing a fast algorithm to find a bidder-optimal core point in a general combinatorial auction (Sections~\ref{sec:water-fill}), and also the experimental problem of deploying our algorithm in the highly time sensitive sponsored search ad auction application with rich ads (Section~\ref{sec:experiment}). For the latter, we use Microsoft Bing ad auction models and data.
\vspace{1mm}

\noindent{\textbf{{Theoretical results: Fast core pricing rule.}}}
Our main result is a fast and simple deterministic algorithm for computing bidder-optimal core payments, given access to an oracle that finds a welfare-optimal allocation for (truncated)  profiles of buyer valuations. In contrast to prior work that find specific bidder optimal core points~\citep{day2007fair, day2012quadratic}, we only find a feasible bidder optimal point; in return, our algorithm provably makes almost linear calls to the oracle.
The algorithm also reveals some specific geometric structures of the core payments polytope.

\medskip
\textbf{Theorem (informal):} \emph{Given oracle access to a welfare-optimal algorithm (for truncated values),  there is a deterministic algorithm  computing an $\epsilon$-approximate bidder-optimal core point with $O(n \log (n/\epsilon))$ oracle evaluations and in time $O(n^2\log(n/\epsilon))$, where $n$ is the number of bidders.}
\medskip

As discussed earlier, the set of core payments (or equivalently utilities) is a polytope, determined by the
(exponentially many) constraints that no subset of bidders can
simultaneously improve their utilities.  Intuitively, the algorithm
proceeds by exploring this polytope through making multiple calls to the oracle.  Starting from an arbitrary point
in the core, the algorithm attempts to increase bidder utilities,
which corresponds geometrically to following a positively-oriented ray
until hitting a facet of the polytope.  It then determines whether
there exists a subset of the bidders whose utilities can be increased
while remaining in the core, and if so continues to follow an
appropriate ray.  This process repeats until a bidder-optimal core
point is reached.


This geometric intuition corresponds to a water-filling algorithm, in
which the utilities of agents are simultaneously increased and frozen
as constraints become tight.  Implementing this approach requires an
efficient test for inclusion in the core polytope, as well as a method
for determining which bidders are involved in a tight core constraint.
It turns out that both questions can be expressed as queries to the
welfare-optimization oracle, with appropriately truncated valuation
profiles.  This result involves an analysis related to the geometry of the core
constraints.
Given these tests, finding a tight core constraint can be done using
binary search (which is what generates the logarithmic factors in our
runtime).  The algorithm terminates after at most $n$ iterations,
since each iteration will freeze the payment of at least one bidder. Notice that this result is naturally related to the well-known
equivalence of optimization and separation, and indeed can be
interpreted as exploiting the geometry of the core to
implement and employ a specially-tailored separation oracle.

Compared to the VCG auction, our algorithm finds payments nearly as
quickly (specifically, with only a logarithmic factor more oracle
calls), and produces higher revenue and a more fair outcome on the
same profile of bids.  Notably, by varying the
direction of binary search in the core polytope, our algorithm can be
parameterized to favor different core outcomes.  For example, we can
maximize the minimum utility enjoyed by any winner (subject to being
in the core), or we can attempt to equalize, across all winners, the
ratio between their utility and the utility they would obtain in the
VCG outcome. We briefly explain this degree of flexibility in Section~\ref{sec:discussion}.

The main drawback of our proposed core pricing algorithm over VCG is the lack of truthfulness (similar to GSP). Computational guarantees are under worst-case and hence do not degrade with
strategic play. For revenue (or even welfare) objective, a more principled way is considering outcomes at equilibria of the auctions. As we discuss in Section~\ref{sec:discussion} and further in Section~\ref{sec:incentive-core} of the supplementary material, our
auction has a natural full information $\epsilon$-Nash equilibrium in which players
play truncation strategies, i.e., shade all values by an additive
constant.  This equilibrium, which has been studied in the literature on core selecting auctions~\citep{milgrom2007package,day2007fair} is welfare-optimal and has higher revenue
than VCG's dominant strategy equilibrium.  Natural bidder dynamics can converge to this equilibrium, specially for the sale of ad space problem, as we explain in Section~\ref{sec:incentive-core} of the supplementary material. We also shed insights on a few future directions to reason about the revenue at equilibrium, e.g., using the computational approach introduced in \citealp{bosshard2017computing,lubin2009quantifying,bunz2018designingEC} to approximate the Bayes Nash Equilibrium (BNE) of core selecting auctions.

%
%

\vspace{1mm}

\noindent{\textbf{{Experimental results: Sale of ad space problem.}}} We conduct a numerical study on using our core selecting auction for the sale of ad space problem.
%
We test the performance of our proposed core auction using the data collected from Microsoft Bing ad auction for text ads.  Our dataset consists of around $20,000$ auctions, each auction containing at most $23$ advertisers, with configurations ranging from the basic three lines up to fourteen lines.
%
%
We compare the performance of our core auction against VCG, minimum revenue core auction of \cite{day2007fair}, quadratic payment rule of \citealp{day2012quadratic}, and a linear programming based method that finds a minimum revenue point in the core by running the volumetric-barrier based cutting plane method of \citealp{vaidya1989new,vaidya1996new}. We also compare with two ad-hoc variants of GSP: (i)  GSP that generates the optimal allocation and charges each ad a price equal to the bid of the next ad in the ordering, and (ii) GSP that uses a greedy allocation -- ignoring ad sizes -- with the same payment rule.  

All of the above algorithms except the GSP auction with greedy allocation require an allocation oracle that finds the optimal welfare (potentially with truncated valuation profiles).  We chose to implement such an optimal oracle.  We show this oracle has an implementation that runs in time polynomial in the number of ads and amount of ad space. See Section~\ref{sec:ad-space-allocation} for more details. 

Our numerical study indicates the existence of complementarities in our setting, which is induced by variable length ads and the fact that we have a fixed amount of space for ads to be shown. See Remark~\ref{remark:complement} for an experimental justification, and Example~\ref{example:complement} for a theoretical justification. As a result, VCG leads to significantly lower revenue than all of core auctions. Considering the reported bids as a proxy for true valuations and analyzing the advertisers at the full information Nash equilibrium of bidder optimal core auctions (see Section~\ref{sec:bid-collection} and Section~\ref{sec:incentive-core} in the supplementary material for more details), our experimental results suggest that our core auction obtains $26\%$ more revenue than VCG, and minimum revenue core auctions obtain $15\%$ more revenue than VCG. Moreover, our core auction almost matches the revenue of GSP and obtains $5\%$ improvement over GSP with greedy allocation.

The sale of ad space problem is highly time sensitive (as Bing runs around 9.6 billion auctions in a month, which is on average $3.7$ auctions every millisecond), and both running time and number of query calls to an optimal welfare allocation oracle matter when measuring the performance of an auction in this domain. We use VCG as a benchmark to compare the running time and query complexity of different methods. Our numerical experiments suggest that our fast core pricing rule yields an improvement in terms of speed (at least $6$-$10$ times faster) over the other core pricing methods we test~\citep{day2007fair,day2012quadratic} and the LP based method for finding a minimum revenue core point using Vaidya's algorithm~\citep{vaidya1989new,vaidya1996new}. Moreover, the running time of our core pricing algorithm is no more than $7.5$ times the running time of VCG on average based on our numerical experiments, which we believe is still in an acceptable speed range for the sale of ad space application. As expected, both ad-hoc variants of GSP have faster running time compared to all core pricing methods and VCG (as they need at most one call to the optimal welfare allocation oracle), but they lack any economic groundings. 

In terms of query complexity, our core pricing algorithm makes far fewer calls to this oracle compared to other core pricing methods. Our numerical experiments suggest that the number of oracle calls of both of the core pricing methods in \cite{day2007fair} and \cite{day2012quadratic} are almost $2.1$ times of our algorithm, and Vaidya's algorithm makes almost $5.9$ times number of calls to the oracle compared to our algorithm. At the same time, our algorithm makes only $3.1$ times the number of calls of the VCG auction (which theoretically speaking makes exactly $n+1$ calls to the oracle, where $n$ is the number of bidders). We believe this query complexity positions our algorithm in an acceptable range to be used for the sale of ad space application.



\subsection{Further Related work}
\label{sec:related-work}
The \emph{ascending proxy auction} was proposed by~\citealp{ausubel2002ascending} as an alternative to VCG to resolve practical  issues such as low revenue. This ascending auction terminates at a bidder optimal core outcome. This line of research was further developed by considering the general notion of core-selecting package auctions~\citep{day2007fair,day2008core,day2012quadratic}. These auctions are used extensively in the public sector where it is of great concern that auctions are efficient and renegotiation-proofs (e.g., see \citealp{ausubel1999optimality}).

Bidder optimal core outcomes are implementable at equilibrium, but are not truthful.  Payment rules that pick particular points in the core have been proposed in the literature to mitigate this issue. 
For example, the difference between a final payment and the VCG payment represents a measure of ``residual incentive to misreport''. \citealp{day2007fair} therefore proposed the minimum-revenue point and \citealp{day2012quadratic} proposed the closest-to-VCG point, and they showed these are the bidder optimal core points that minimize $\ell_1$-norm and $\ell_2$-norm of this difference respectively. An alternative practical core payment rule proposed by \citealp{erdil2010new} considers robustness with respect to the submitted bids. More recently, computational approaches to approximate the BNE of different core payments have been proposed~\citep{bosshard2017computing,lubin2009quantifying,lubin2015new,bunz2015faster,bunz2018designingEC,bunz2018designing}, which was further used for an algorithmic search for ``good" core-selecting rules~\citep{bunz2018designingEC,bunz2018designing}


The winner determination problem for truncated values is effectively a separation oracle for the core polytope. Therefore, given access to this separation oracle, one can  optimize convex objectives (e.g., revenue or $\ell_2$-distance from VCG) exactly over the core polytope in polynomial time using the Ellipsoid algorithm~\citep{grotschel1981ellipsoid}. However, this approach is rather slow in practice. Core pricing rules of \citealp{day2007fair} and \citealp{day2012quadratic} take a  different approach: they  keep track of a small relaxation to core polytope that iteratively becomes tighter, and at each iteration solve a linear programming or quadratic programming (based on the objective function) to eventually find a feasible core point. To update the relaxation, these methods rely on a \emph{core constraint generation} family of algorithms~\citep{day2007fair,day2012quadratic,bunz2015faster}, simplest of which send queries to the core separation oracle to find the most violated core constraint and add it to the current relaxation of the core to make it tighter. 

Given the separation oracle for the core polytope, another approach is using ``cutting plane methods" to utilize this oracle efficiently in order to find a minimum revenue core point much faster than the Ellipsoid method. One such option is Vaidya's algorithm~\citep{vaidya1989new,vaidya1996new}, which is essentially an interior point method using the volumetric barrier (see \citealp{bubeck2015convex} for details). Moreover, it finds an $\epsilon$-close point to a minimum revenue core point with oracle complexity of $O\left((n/\epsilon)\log n\right)$ and extra computation of $O\left(n^4/\epsilon\right)$. Alternatively, we can use the more recent fast cutting-plane methods for the feasibility linear programs of \citealp{lee2015faster} in a black-box fashion, which then finds an $\epsilon$-close minimum revenue core point with high probability, using $O\left(n^{(1+\epsilon)} \log^{O(1)}(n/\epsilon)\right)$ oracle evaluations in expectation and additional time $O\left(n^3\log^{O(1)}(n/\epsilon)\right)$. Our algorithm has several advantages over this general purpose approach in terms of running time, oracle complexity, usage of randomness, and practicality. See Section~\ref{sec:discussion} for a comprehensive theoretical comparison, and Section~\ref{sec:practical-considerations} in the supplementary material for a list of practical obstacles to use this approach in our sale of ad space application.

Numerous work have proposed core auctions specifically to address known drawbacks of VCG auctions.
The Bayesian setting was studied by~\citealp{ausubel2010core}, and more recently by~\citealp{sano2012non} and~\citealp{goeree2016impossibility}, providing theoretical and experimental evidence that the revenues and efficiency from core-selecting auctions improve as correlations among bidders' values increase, while the revenues from the VCG auction worsen.  Another issue with VCG is the lack of revenue monotonicity: adding bidders or increasing bids can potentially decrease the seller's revenue~\citep{goel2014revenue,lamy2010core,rastegari2011revenue,beck2009revenue}. 
\citealp{beck2009revenue} and \citealp{lamy2010core} proposed revenue-monotone core selection. 


Sponsored search has been a central research topic in the past decade~\citep{edelman2007internet,aggarwal2006truthful,wilkens2016mechanism,wilkens2017gsp}. The closest in the literature to our sale of ad space application is the work of~\cite{cavallo2017sponsored} that considered the algorithmic question of finding allocation and pricing for rich ads, but did not study core pricing specifically.



\section{Preliminaries}
We describe our algorithm and theoretical results in a general model of combinatorial auctions first, then return to a restricted model tailored to the ad auction setting in Section~\ref{sec:experiment}.
Consider a combinatorial auction with $n$ bidders and $m$ items. Let $N\triangleq \{1,2,\ldots,n\}$ be the set of bidders.  The auction asks each bidder to declare a valuation function, which assigns a value to each subset of the items.  We will write $b_i$ for the valuation function submitted by agent $i$.  We assume that valuation functions are normalized so that $b_i(\emptyset) = 0$, and all values are in $[0,1]$.  An allocation is an assignment of item bundles to agents, $\{x_i\}_{i\in N}$, such that $x_i\in 2^{[m]}$ for each $i$ and $x_i\cap x_j=\emptyset$ for all $i \neq j$. 

Given the bid functions $\{b_i\}_{i \in N}$, where $b_i:2^{[m]}\rightarrow \mathbb{R}$ maps each bundle of items to a real number,  the auction will return an \emph{outcome}. An outcome is consisting of an allocation $\{x_i\}_{i \in N}$ and a payment $p_i \geq 0$ for each bidder. Since our focus is the computational problem of finding core payments, and not buyer incentives, we do not differentiate between true and declared valuations in our notation.  With this convention, $\pi_i$ denotes the resulting utility of bidder $i$, so that $\pi_i = b_i(x_i) - p_i$.  We write $\pi_0$ for the seller's revenue.
Throughout this paper, we assume that the auction's allocation rule does not allocate sets of zero value. We say that a feasible allocation $\{x^*_i\}_{i\in N}$ is welfare-maximizing if and only if $\{x^*_i\}_{i\in N}\in\underset{x: x_i\cap x_j=\emptyset, i\neq j}{\textrm{argmax}}\sum_{i\in N }b_i(x_i)$. For any feasible allocation $\{x_i\}_{i\in N}$, let $W(\{x_i\}_{i\in N})\triangleq\{i\in N: x_{i}\neq \emptyset\}$ be the corresponding winners of this allocation.   Let $w(S,\{b_i(\cdot)\}_{i\in N})$ be the maximum welfare of coalition $S\subseteq N$ with respect to the bids $\{b_i\}$. That is, $w(S,\{b_i(\cdot)\}_{i\in N})$ is the maximum welfare of any allocation that assigns items only to agents in $S$. Unless noted otherwise, we write $w(S)\triangleq w(S,\{b_i(\cdot)\}_{i\in N})$ for notational convenience. 
\begin{definition}[\textbf{Winner Set}]
A subset $W\subseteq N$ of bidders is a \emph{winner set with respect to bids $\{b_i(\cdot)\}_{i\in N}$} if and only if there exists a feasible allocation $\{x^*_i\}_{i\in N}$ such that 
\begin{equation}
\{x^*_i\}_{i\in N}\in\underset{x: x_i\cap x_j=\emptyset, i\neq j}{\operatorname{argmax}}\sum_{i\in N }b_i(x_i)
\end{equation}
and  $W=\{i\in N: x^*_{i}\neq \emptyset\}$. Let $\mathcal{W}(\{b_i(\cdot)\}_{i\in N})$ be the set of winner sets with respect to $\{b_i(\cdot)\}_{i\in N}$.
\end{definition}

\begin{definition}[\textbf{Core}]
\label{def:core}
A vector of non-negative utilities $\{\pi_i\}_{i\in N\cup\{0\}}$ is said to be \emph{in the core with respect to the submitted bids $\{b_i(\cdot)\}_{i\in N}$} if no blocking coalition exists, that is, no group of bidders plus the seller can deviate to simultaneously improve their outcomes, including the seller's revenue:
\begin{equation}
\label{eq:core}
\forall	S\subseteq N: \pi_0\geq w(S)-\sum_{i\in S}\pi_i.
\end{equation}
Similarly, an outcome $\{(x_i,p_i)\}_{i\in N}$ is said to be \emph{in the core} if its corresponding vector of utilities $\{\pi_i\}_{i\in N\cup\{0\}}$ is in the core, where $\pi_i=b_i(x_i)-p_i$ for $ i\in N$ and $\pi_0=\sum_{i\in N}p_i$.
\end{definition}
It is easy to see that any core outcome $\{(x_i,p_i)\}_{i\in N}$  is welfare-maximizing, as for coalition $N$ we have $\sum_{i\in N}b_i(x_i)=\pi_0+\sum_{i\in N}\pi_i\geq w(N)$. Therefore, one can rewrite Equation~\ref{eq:core} as
\begin{equation}
\label{eq:core-alter}
 \forall S\subseteq N: w(N)-\sum_{i\in N}\pi_i\geq w(S)-\sum_{i\in S} \pi_i~.
\end{equation} 

We have defined the core with respect to utilities and with respect to outcomes.  The following lemma shows that given a core point in utility space, it is possible to reconstruct core allocations and payments.  This motivates us to focus on the problem of computing core points in utility space. 
\begin{lemma}
Given a core point $\{\pi_i\}_{i\in N}$ and any maximum welfare allocation $\{x^*_i\}_{i\in N}$,  let $p^*_i=b_i(x^*_i)-\pi_i$ for all $i\in N$. Then $\{(x^*_i,p^*_i)\}_{i\in N}$ will be an outcome in the core.   
\end{lemma}
\proof{Proof.}
By Definition~\ref{def:core}, as long as payments are non-negative, we have the outcome $\{(x^*_i,p^*_i)\}_{i\in N}$ to be in the core. So, we only need to show for every $i\in N$,  $\pi_i\leq b_i(x^*_i)$. By looking at the core constraint for coalition $N\setminus\{i\}$, we have:
\begin{align*}
w&(N)-\sum_{j\in N}\pi_j\geq w(N\setminus\{i\})-\sum_{j\in N\setminus\{i\}}\pi_j\Rightarrow\nonumber\\
 \pi_i&\leq w(N) - w(N\setminus\{i\})=\sum_{j\in N}b_j(x^*_j)-w(N\setminus\{i\})\nonumber\\
 &\leq \sum_{j\in N}b_j(x^*_j)-\sum_{j\in N\setminus\{i\}}b_j(x^*_j)=b_i(x^*_i)~,
\end{align*} 
where the last inequality holds as $\{x^*_j\}_{j\in N\setminus\{i\}}$ is a feasible allocation for bidders $N\setminus\{i\}$. \hfill\Halmos
\endproof

Our goal is to compute bidder optimal core points, i.e., points in the core that are not dominated by any other point in the core with respect to bidder utilities. We relax this condition by defining $\epsilon$-bidder optimal core points, i.e., core points that are $\epsilon$-close to being bidder Pareto optimal. 
\begin{definition}
\label{def:epsbidderoptimal}
 A vector of utilities $\{\pi_i\}_{i\in N}$ is said to be \emph{$\epsilon$-bidder optimal} for $\epsilon>0$, if 
\begin{itemize}
\item $\{\pi_i\}_{i\in N}$ is in the core with respect to the submitted bids $\{b_i(\cdot)\}_{i\in N}$, and
\item There exists no other core point $\{\pi'_i\}_{i\in N}$ such that for every bidder $i\in N$,  $\pi'_i\geq \pi_i$, and for at least one bidder $j$, $\pi'_j>\pi_j+\epsilon$.
\end{itemize}
Similarly, an outcome $\{(x_i,p_i)\}_{i\in N}$ is called $\epsilon$-bidder optimal if its utility vector is $\epsilon$-bidder optimal. 
\end{definition}

Observe that the set of core points forms a polytope in utility space, described by the (exponentially many) constraints of the form in inequality~\eqref{eq:core-alter}.
To understand the structure of this polytope, it is instructive to consider the core constraint in Definition~\ref{def:core} whose right-hand side takes the maximum value. Therefore, we define the notion of \emph{maximum binding core constraints.}
\begin{definition}
\label{def:mostbinding}
Given a vector of utilities $\{\pi_i\}_{i\in N}$, a subset $S\subseteq N$ is a \emph{Maximum Binding Core Constraint (MBCC)} if and only if $\displaystyle S\in \underset{S'\subseteq N}{\operatorname{argmax}} ~w(S') -\sum_{i\in S'}\pi_i$.

\end{definition}
\begin{definition}
\label{definition:ep-tight} 
Given a vector of utilities $\{\pi_i\}_{i\in N}$, a subset $S\subseteq N$ 
is an \emph{$\epsilon$-tight core constraint} (or a tight core constraint when $\epsilon=0$) if and only if
\begin{equation}
 w(S) -\sum_{i\in S}\pi_i\leq w(N)-\sum_{i\in N}\pi_i\leq w(S) -\sum_{i\in S}\pi_i+\epsilon~.
\end{equation}
\end{definition}

%

\section{The Water-Filling Algorithm}
\label{sec:water-fill}
In brief, our algorithm is a water-filling method that starts from an arbitrary point in the core (e.g., pay-your-bid, where all agents receive utility $0$) and then at each iteration does the following: 
\begin{itemize}
\item \emph{Finding a feasible direction}: It first finds a subset of bidders such that if we increase all of their utilities uniformly by a small amount, we still remain in the core.
\item \emph{Uniform utility increase}: It then increases the utilities for those bidders uniformly, until it hits a facet of the core (approximately).
\item \emph{Checking termination condition/tepeat:} Finally, it checks whether there exists any remaining subset of bidders who can increase their utilities  and still remain in the core. If so, the algorithm iterates. Otherwise it terminates.
\end{itemize} 
 
 See Figure~\ref{fig:algorithm} for an overview of the algorithm, and its geometric interpretation. 
 \begin{figure}[htb]
\begin{center}
figures \includegraphics[scale=0.7]{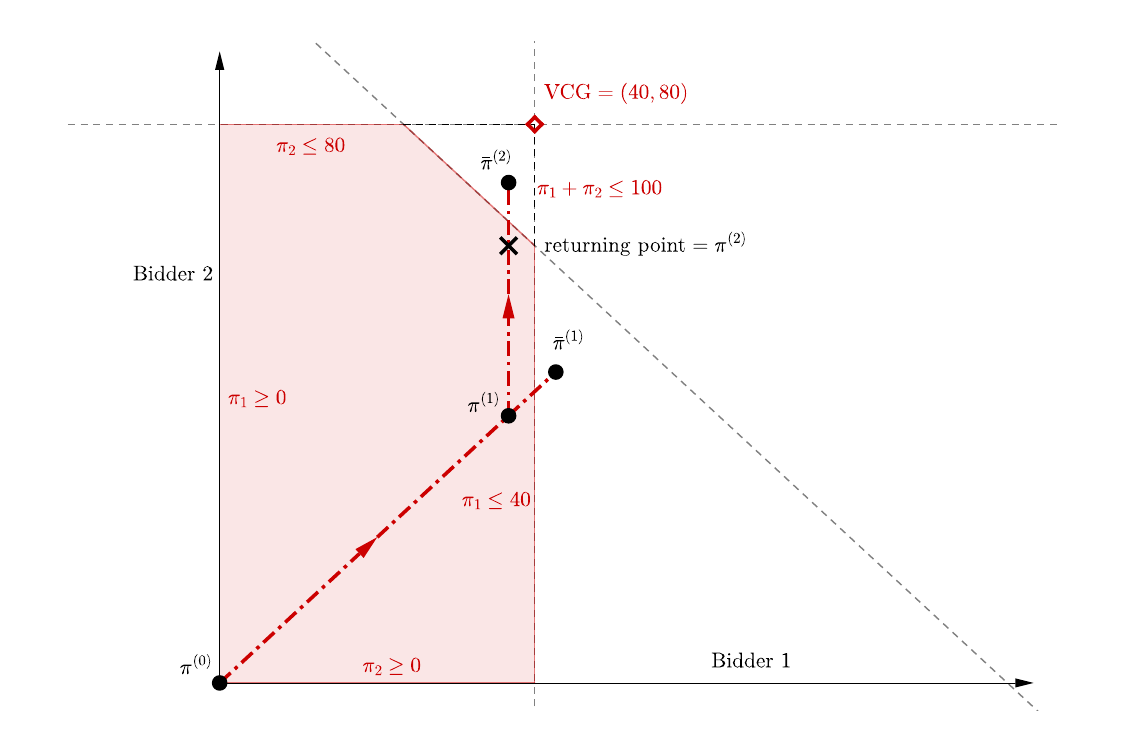}
\caption{Water-filling algorithm for the following example in~\citealp{day2007fair}. $2$ items A and B, $5$ bidders; bids are as follow (each bidder submits one bid): $b_1(A)=60$, $b_2(B)=100$, $b_3(AB)=60$, $b_4(A)=20$, and $b_5(B)=20$. The filled area is the core-polytope, and the arrow shows the path that algorithm follows. Note that $S_0=\{1,2,3,4,5\}$, $T_0=\{1,2\}$, $S_1=\{1,2\}$, $T_1=\{2,4\}$, $S_2=\{2\}$, $T_2=\{3\}$ and $S_3=\emptyset$.}
\label{fig:algorithm}
\end{center}
\end{figure} 


\subsection{Winner Determination Oracle Model}
\label{sec:windet}
In this paper, we focus on the \emph{winner determination oracle model} as our model of computation. More accurately, suppose we have access to a winner determination oracle under sincere strategies, where a sincere strategy is a truncation of $\{b_i(.)\}$ by offsets $\{\pi_i\}$, that is, $\{\max(b_i(.)-\pi_i,0)\}$. The oracle accepts submitted bids and the required truncations as input, and then simulates truncated bids to find a maximum welfare allocation. To make this more concrete, we define the oracle $\texttt{WIN-ORAC}$.
\begin{definition}
\label{def:win-det}
 Let $\texttt{WIN-ORAC}$ be a black-box oracle with this input-output relation:
\begin{itemize}

\item \texttt{Input:} submitted bids $\{b_i(.)\}_{i\in N}$ and required truncations $\{\pi_i\}_{i\in N}.$
\item \texttt{Output}: max welfare with truncated bids $:= w(N,\{\max(b_i(.)-\pi_i,0)\}_{i\in N})$, and a winning set, i.e., $S\in \mathcal{W}(\{\max(b_i(.)-\pi_i,0)\}_{i\in N})$ (or \texttt{NULL} if no such set exists).
\end{itemize}
\end{definition}


\subsection{The Algorithm}
\label{sec:alg}
   Here is the full description of our algorithm (Algorithm~\ref{alg:bidderopt}). It uses the subroutine \texttt{Core-Search} (Procedure~\ref{alg:binary search}) at each iteration, which is basically binary-search to find the next feasible subset of coordinates/bidders quickly, and it only requires $\tilde{O}(1)$ number of calls to the oracle $\texttt{WIN-ORAC}$.
  
  \begin{algorithm}[htb]
	\caption{Water-filling for finding an $\epsilon$-bidder optimal core point}
	 \label{alg:bidderopt}
\KwIn{submitted bids $\{b_i(.)\}$, set of bidders $N$, and $\epsilon>0$}
	 \vspace{1mm}
	 
 		\textbf{initialize} $\pi^{(0)}_i=0$ and $\bar{\pi}^{(0)}_i=0$ for all $i\in N$
 		 ~~\tcp{utilities in pay-your-bid auction.}
 	
 		 $t\leftarrow 0$, $S_0\leftarrow N$
 		 
 		 \While{$S_t\neq \emptyset$}{
 		 $T_t$ $\gets$ any set in $\mathcal{W}(\{\max(b_i(.)-\bar{\pi}^{(t)}_i,0)\})$ \tcp{needs a query call to \texttt{WIN-ORAC}.}
 		 
 		 $S_{t+1}\gets S_t\cap T_t$
 		 
 		 \If{$S_{t+1}\neq\emptyset$}{
 		 Run procedure \texttt{Core-Search}$(\{b_i(.)\}, \pi^{(t)},S_{t+1},\epsilon)$ to return $\{\bar{\pi_i}\}_{i\in N}$ \& $\{\ubar{\pi}_i\}_{i\in N}$
 		 
 		 $\pi^{(t+1)}_i\leftarrow \ubar{\pi}_i~,~ i\in N$
 		 
 		 $\bar{\pi}^{(t+1)}_i\leftarrow \bar{\pi}_i~,~ i\in N$
 		 
 		 $t\gets t+1$	
 		 }
 		 
 		 \textbf{return} $\pi^{(t)}$
 		 
 		 }
\end{algorithm}

%
%

\begin{myprocedure}
\caption{\texttt{Core-Search} for water-filling algorithm}
 \label{alg:binary search}
 \KwIn{submitted bids $\{b_i(.)\}_{i\in N}$, core point $\pi$, subset of bidders $S\subseteq N$, and $\epsilon>0$}
 \vspace{1mm}
 
 \textbf{initialize} $\Delta_l=0$ and $\Delta_h=1$.

 \While{$\Delta_h-\Delta_l>\frac{\epsilon}{\lvert S\rvert}$}{
 
  $\Delta\leftarrow\frac{\Delta_l+\Delta_h}{2}$~~~~~~\tcp{do binary-search to find $\Delta_l$ and $\Delta_h$.}
  
Let $\tilde{\pi}_i=\pi_i+\Delta$ for $i\in S$, $\tilde{\pi}_i=\pi_i$ for $i\in N\setminus S$, and $\tilde{\pi}_0=w(N)-\sum_{i\in N}\tilde{\pi}_i$

\uIf{$\tilde{\pi}_0\geq w\left(N,\{\max(b_i(.)-\tilde{\pi}_i,0)\}\right)$}{
\tcc{requires one query call to oracle \texttt{WIN-ORAC}.}

$\Delta_l\leftarrow\Delta$.
}
\Else{
$\Delta_h\leftarrow\Delta$
}
}
\vspace{-5mm}
  \textbf{return} $\{\bar{\pi_i}\}_{i\in N},\{\ubar{\pi}_i\}_{i\in N}$, where $\begin{cases}
    \bar{\pi_i}=\pi_i+\Delta_h~\textrm{and}~\ubar{\pi}_i=\pi_i+\Delta_l, & \text{if $i\in S$}.\\
    \bar{\pi_i}=\pi_i~\textrm{and}~\ubar{\pi}_i=\pi_i, & \text{otherwise}.
 \end{cases}$
\end{myprocedure}

%


The iterations of the algorithm are indexed by $t$.  The algorithm maintains a current core point $\pi^{(t)}$ (initially the origin) and a collection of active bidders, $S_t$, whose utilities can potentially be increased (initially all bidders).  On each iteration, the algorithm updates the set of active bidders $S_t$ by finding and removing a set $S_t\setminus T_t$ of bidders whose utilities cannot be increased without violating a core constraint (Finding the next feasible set might take several iterations, without changing the point $\pi^{(t)}$, for reasons that will become clear in the proof).  It then applies binary search along the ray of points consisting of uniform increases to the utilities of all agents in $S_{t+1}$, described in subroutine $\texttt{Core-Search}$.  This finds a pair of points $\{\ubar{\pi}_i\}$ and $\{\bar{\pi}_i\}$, where $\ubar{\pi}$ is in the core, $\bar{\pi}$ is outside the core, and the points are within $\epsilon$ of each other in $\ell_1$ distance, i.e.
\begin{align*}
&w(N)-\sum_{i\in N}\ubar{\pi}_i\triangleq\ubar{\pi}_0\geq w\left(N,\{\max(b_i(.)-\ubar{\pi}_i,0)\}\right)~,~~\triangleright~\textrm{feasible core point.}\\
&w(N)-\sum_{i\in N}\bar{\pi}_i\triangleq\bar{\pi}_0 < w\left(N,\{\max(b_i(.)-\bar{\pi}_i,0)\}\right)~,~~\triangleright~\textrm{infeasible (out of core).}\\
&\sum_{i\in N}\bar{\pi}_i-\sum_{i\in N}\ubar{\pi}_i\leq \epsilon~.
\end{align*}	
 The algorithm uses $\ubar{\pi}$ as its updated core point, and uses $\bar{\pi}$ to update the set of active bidders in the subsequent iteration, by finding the set of winners for sincere bids $\{\max(b_i(.)-\bar{\pi}_i,0)\}$, i.e. $T_{t+1}$.  Once all bidders are frozen, the algorithm returns the current core point. An illustration of each iteration and how sets $S_t$ and $T_t$ are set at each iteration can be seen in Figure~\ref{fig:algorithm}, which describes simulating our algorithm on the example in~\cite{day2007fair}.

\subsection{Proof of Correctness and Running Time}
\label{sec:proof}
The main idea behind the proof of correctness of the algorithm is the following simple observation. Suppose $\pi$ is a point in the core and $S$ is a tight core constraint with respect to $\pi$. Note that there always exists at least one tight core constraint, as the constraint for coalition $N$ is always tight. Now, one is allowed to increase $\pi_i$ by a small amount and still  have a core point as long as bidder $i$ is participating in \emph{every} tight core constraint $S\subseteq N$. This is true because the change in the left-hand side and right-hand of Equation~\ref{eq:core-alter} will be the same for all currently tight constraints, which are the only candidates for violation after the small change. Inspired by this observation, the  algorithm starts from a point in the core and increases utilities of nested subsets of bidders uniformly at each iteration, until  no bidder in the intersection of all tight (and almost tight) core-constraints exists.

Given this observation, here is the intuition behind the correctness of the algorithm. As mentioned earlier, Algorithm~\ref{alg:bidderopt} keeps track of $\{\pi_i^{(t)}\}_{i\in N}$ (inside the core), and $\{\bar{\pi}_i^{(t)}\}_{i\in N}$ (outside of the core), while these two points are always $\epsilon$-close in $\ell_1$-norm distance. As a result, they help the algorithm to find the next subset of bidders $S_{t+1}\subset S_{t}$, with the property that the algorithm can potentially increase their corresponding utilities uniformly. Furthermore, for a fixed run of the algorithm consider sequence $\emptyset=\mathcal{G}_0\subset\mathcal{G}_1\subset\mathcal{G}_2\subset,\ldots$, where $\mathcal{G}_{t}\triangleq\{T_1,\ldots,T_{t-1}\}$. By closeness of the two points, it can be shown that $\mathcal{G}_t$ is always a collection of $\epsilon$-tight core constraints with respect to $\pi^{(t)}$. This sequence acts as a \emph{certificate of correctness} for the algorithm: at each iteration $t$ the algorithm ensures that bidders in $S_{t+1}=\bigcap_{t'=0}^{t}T_t$ are indeed the subset of bidders appearing in \emph{every} constraint of $\mathcal{G}_{t+1}$, and hence those which could have increased their utility. So, at termination it ensures that $S_{t+1}$ is empty, and therefore there exists no bidder that appears in all $\epsilon$-tight core constraints.

More concretely, we prove the correctness of the algorithm in two steps. We first show upon termination the algorithm outputs a point in the core that is $\epsilon$-bidder optimal, and then we show that the algorithm terminates in at most $\lvert N \rvert=n$ iterations. To do so, we start by proving the following lemma, which is crucial in understanding how our algorithm works. Notably, the first bullet of this lemma is inspired by similar results in~\citealp{day2007fair}.
\begin{lemma}
\label{lemma:most-binding}
Fix a vector of utilities $\{\pi_i\}_{i\in N}$. Let $S\subseteq N$. The following are true:
\begin{itemize}
\item $S$ is a maximum binding constraint if it is also the set of winning bidders for some maximum welfare allocation under bids $\{\max(b_i(.)-\pi_i,0)\}_{i\in N}$.
\item Suppose we have $\forall i\in S: \pi_i>0$. Then $S$ is a maximum binding constraint only if it is the set of winning bidders for some maximum welfare allocation under bids $\{\max(b_i(.)-\pi_i,0)\}_{i\in N}$.
\end{itemize}
\end{lemma}
\proof{Proof.}
Let $S'$ be the set of winners for a maximum welfare allocation  $x'$ under truncated bids $\{\max(b_i(.)-\pi_i,0)\}_{i\in N}$. Note for all $i\in S'$, $b_i(x'_i)> \pi_i$ (otherwise, allocation $x'$ gives items for free to some bidder $j$, because  $x'_j\neq \emptyset ~\&~\max(b_j(x'_j)-\pi_j,0)=0$ as $b_j(x'_j)\leq\pi_j$). Let $S$ be a maximum binding constraint and $x^{S}$ be the maximum welfare allocation restricted to bidders $S$. We have
\begin{align}
w(S)-\sum_{i\in S}\pi_i&=\sum_{i\in S}(b_i(x^{S}_i)-\pi_i)\leq\sum_{i\in N}\max(b_i(x^{S}_i)-\pi_i,0)\nonumber\\
&\leq \sum_{i\in N}\max(b_i(x'_i)-\pi_i,0)=\sum_{i\in S'}(b_i(x'_i)-\pi_i)\leq w(S')-\sum_{i\in S'}\pi_i~.
\end{align}
Therefore $S'$ is also a maximum binding constraint.

Next, suppose $S$ is a maximum binding constraint. Let $x^S$ be the allocation that maximizes the welfare restricted to bidders in $S$. Note that in such an allocation, all the bidders in $S$ will be winners, i.e. $x^S_i\neq \emptyset$ for $i\in S$ (because otherwise there exists $S'\subset S$ such that $w(S)=w(S')$, and therefore $w(S)-\sum_{i\in S}\pi_i<w(S')-\sum_{i\in S'}\pi_i$ which is a contradiction). Now let $x'$ be a maximum welfare allocation for truncated bids, and let $S'$ be the set of winners under such an allocation. Then:
\begin{align}
\sum_{i\in N}\max(b_i(x'_i)-\pi_i,0)&=\sum_{i\in S'}(b_i(x'_i)-\pi_i)\leq w(S')-\sum_{i\in S'}\pi_i\leq w(S)-\sum_{i\in S}\pi_i\nonumber\\
&=\sum_{i\in S}(b_i(x^S_i)-\pi_i)\leq\sum_{i\in S}{\max(b_i(x^S_i)-\pi_i,0)}=\sum_{i\in N}{\max(b_i(x^S_i)-\pi_i,0)}~.
\end{align}
Therefore $x^S$ is an optimal allocation under $\{\max(b_i(.)-\pi_i,0)\}_{i\in N}$, and $S$ is its winner set. \hfill\Halmos
\endproof

   Now, using Lemma~\ref{lemma:most-binding}, we can prove the following invariants of our algorithm: $\pi^{(t)}$ is a core point for each $t$, and $\{\mathcal{G}_t\}$ is a collection of $\epsilon$-tight core constraints with respect to $\pi^{(t)}$.
\begin{proposition} 
\label{lemma:invariant}
Given submitted bids $\{b_i(.)\}_{i\in N}$ and $\epsilon>0$, there exists a finite sequence of collections $\emptyset=\mathcal{G}_0\subset\mathcal{G}_1\subset\ldots\subset \mathcal{G}_{T+1}$, such that the following invariants hold at each iteration $t$ of  Algorithm~\ref{alg:bidderopt}:
\begin{enumerate}

\item[(1)] $\{\pi^{(t)}_i\}_{i\in N}$ is always in the  core, and $\{\bar{\pi}^{(t)}_i\}_{i\in N}$ is outside of the core for $t\geq 1$.

\item[(2)] $S_{t}$ is the subset of bidders that are simultaneously participating in all core constraints included in the collection $\mathcal{G}_{t}\triangleq \{T_0,\ldots,T_{t-1}\}$. Moreover, for $t\geq 1$, $S_{t}\setminus T_{t}\neq \emptyset$ and $\mathcal{G}_t\neq \mathcal{G}_{t+1}$.

\item [(3)] $\mathcal{G}_{t+1}$ is a collection of $\epsilon$-tight core constraints with respect to $\pi^{(t)}$ (as in Definition~\ref{definition:ep-tight}).
\end{enumerate}
\end{proposition}
\begin{proof}{Proof of Part (1):} To prove this part, we use induction. For the base case $t=0$, all-zero vector $\pi^{(0)}$ is always in the core as $w(N)\geq w(S)$ for all $S\subseteq N$. Now suppose $\pi^{(t-1)}$ is in the core. If $\pi^{(t)}=\pi^{(t-1)}$ we are done. So let $\pi^{(t)}\neq \pi^{(t-1)}$. As a result, $\pi^{(t)}=\ubar{\pi}$ in the binary search phase at iteration $t-1$. Therefore, due to termination condition of \texttt{Core-Search},  for all $S\subseteq N$ we have:
\begin{equation*}
w(N)-\sum_{i\in N}\pi^{(t)}_i=\ubar{\pi}_0\geq w\left(N,\{\max(b_i(.)-\pi^{(t)}_i,0)\}\right)\geq w(S)-\sum_{i\in S}\pi^{(t)}_i~,
\end{equation*}
where the last inequity holds because of Lemma~\ref{lemma:most-binding}. So, $\pi^{(t)}$ is a core point. Also, for $t\geq 1$ the point $\bar{\pi}^{(t)}$ is outside of the core, since once binary search stops at iteration $t-1$, for some $S\subseteq N$ we have:
\begin{equation*}
w(N)-\sum_{i\in N}\bar{\pi}^{(t)}_i=\bar{\pi}_0< w\left(N,\{\max(b_i(.)-\bar{\pi}^{(t)}_i,0)\}\right)= w(S)-\sum_{i\in S}\bar{\pi}^{(t)}_i\tag*{\Halmos}
\end{equation*}
\end{proof}
\begin{proof}{Proof of Part (2):} First of all, note that $S_t$ is the subset of bidders that are simultaneously participating in all of $\{T_{t'}\}_{t=1}^{t-1}$, simply because $S_t=\bigcap_{t'=0}^{t-1}T_{t'}$ due to the update rule of $S_t$. Moreover, at iteration $t-1$ (for $t\geq 1$) the algorithm starts from a feasible core point $\pi^{(t-1)}$ and uniformly increases the utilities only for bidders in $S_{t}$, until it reaches to a point $\bar{\pi}^{(t)}$ that is outside of the core. Note that no constraint $S\supseteq S_{t}$ will get violated during this process, simply because the changes in the left hand side and right hand side of these constraints (refer to Definition~\ref{def:core} and inequality~\eqref{eq:core-alter}) are equal, following the fact that $S_{t}=\bigcap_{t'=0}^{t-1}T_{t'}$. In particular, no constraint in $\mathcal{G}_t$ will get violated by $\bar{\pi}^{(t)}$. At iteration $t$,  $T_{t}$ is set to one of the most binding core constraint with respect to $\bar{\pi}^{(t)}$ (due to Lemma~\ref{lemma:most-binding}), and therefore it should be a violated core constraint, because $\bar{\pi}^{(t)}$ is outside of the core. Combining the above arguments, $S_t\setminus T_t\neq \emptyset$. So, $T_t\notin \mathcal{G}_t$ and therefore $\mathcal{G}_{t}\neq\mathcal{G}_{t+1}$.\hfill\Halmos
\end{proof}
\begin{proof}{Proof of Part (3):} To prove this part, we again use induction. For the base case $t=1$, $\mathcal{G}_1=\{T_0\}$ is the set of winning bidders of maximum welfare allocation under original bids $\{b_i(.)\}_{i\in N}$. Therefore, due to Lemma~\ref{lemma:most-binding}, it is also a maximum binding constraint for all-zero vector $\pi^{(0)}$. Moreover, coalition $N$ is always a tight core constraint (as in Definition~\ref{definition:ep-tight}),  and therefore any maximum binding core constraint is also a tight core constraint. Now suppose $\mathcal{G}_t=\{T_0,\ldots,T_{t-1}\}$ is a collection of $\epsilon$-tight core constraints with respect to $\pi^{(t-1)}$. Note that if a constraint is $\epsilon$-tight core constraint at some iteration, it will always remain $\epsilon$-tight, as utilities never decrease and $\pi^{(t)}$ is also in the core. At iteration $t$, $\pi^{(t)}=\ubar{\pi}$ and $\bar{\pi}^{(t)}=\bar{\pi}$, where these parameters are set in binary search phase at the previous iteration $t-1$. Moreover, $T_t$ will be a violated constraint under $\bar{\pi}$ (as in the proof of Part~(2)). Therefore, 
\begin{equation*}
w(N)-\sum_{i\in N}\pi^{(t)}_i\overset{(1)}{\leq} w(N)-\sum_{i\in N}\bar{\pi}^{(t)}_i+\epsilon\overset{(2)}{<} w(T_t)-\sum_{i\in T_t}\bar{\pi}_i+\epsilon\overset{(3)}{\leq} w(T_t)-\sum_{i\in T_t}\pi^{(t)}_i+\epsilon~,
\end{equation*}
where inequality $(1)$ holds as $\sum_{i\in N}\bar{\pi}_i-\sum_{i\in N}\ubar{\pi}_i\leq \epsilon$, inequality $(2)$ holds as $T_t$ is a violated core constraint for $\bar{\pi}^{(t)}=\bar{\pi}$, and inequity $(3)$ holds as for all $i\in N$, $\pi^{(t)}_i=\ubar{\pi}_i\leq \bar{\pi}_i$. So, $T_t$ is also an $\epsilon$-tight core constraint for $\pi^{(t)}$ and $\mathcal{G}_{t+1}$ is a collection of $\epsilon$-tight core constraints with respect to $\pi^{(t)}$. \hfill\Halmos
\end{proof}

\begin{lemma}
\label{lemma:termination}
The Algorithm~\ref{alg:bidderopt} terminates in at most $\lvert N\rvert=n$ iterations.
\end{lemma}
\proof{Proof.}
By using Proposition~\ref{lemma:invariant}, Algorithm~\ref{alg:bidderopt} terminates, as $\mathcal{G}_t\neq\mathcal{G}_{t+1}$ for $t\geq 1$ and there are only finitely many such collections. Moreover, following the fact that $S_t\setminus T_t\neq \emptyset$, one can conclude $S_{t+1}\subsetneq S_{t}$. $S_0=N$, and hence the algorithm terminates in at most $\lvert N\rvert$ iterations. \hfill\Halmos
\endproof

\begin{theorem}
\label{thm:main}
Algorithm~\ref{alg:bidderopt} returns an $\epsilon$-bidder optimal core point. Moreover, it requires $O(n \log(n/\epsilon))$ evaluations of the oracle \texttt{WIN-ORAC} and an additional time $O(n^2 \log(n/\epsilon))$.
\end{theorem}
\proof{Proof.}
Combining Proposition~\ref{lemma:invariant} and Lemma~\ref{lemma:termination}, after at most $\lvert N\rvert=n$ iterations the algorithm terminates. At termination time $T$, $S_{T+1}=\emptyset$ is the subset of bidders that are participating simultaneously in all of the core constraints in the collection $\mathcal{G}_{T+1}$. Also, $\mathcal{G}_{T+1}$ is a collection of $\epsilon$-tight core constraints with respect to $\pi^{(T)}$. Combining those, we conclude that the intersection of all $\epsilon$-tight core constraints with respect to $\pi^{(T)}$ is empty. Now, $\pi^{(T)}$ is in the core and if you increase one of its coordinates, lets say $j$, by more than $\epsilon$ then we know there exists at least one $\epsilon$-tight core constraint $S$ with respect to $\pi^{(T)}$ such that $j\notin S$ and therefore by this change this constraint will be violated. So, by Definition~\ref{def:epsbidderoptimal}, $\pi^{(T)}$ is $\epsilon$-bidder optimal. Moreover, at each iteration $t$ the algorithm uses one query call to \texttt{WIN-ORAC} to find $S_{t+1}$, at most $\log(\frac{\lvert S_{t+1}\rvert}{\epsilon})$ query calls to do the binary search, and at most $\lvert S_{t+1}\rvert\log(\frac{\lvert S_{t+1}\rvert}{\epsilon})$ extra additions during the binary search. So in total it only needs
\begin{align}
\textrm{Total $\#$ of oracle calls}&\leq n+\sum_{t=0}^{T}\log\left(\frac{\lvert S_{t+1}\rvert}{\epsilon}\right)\leq n+\sum_{t=1}^{n}\log\left(\frac{t}{\epsilon}\right)\nonumber\\
&=n+\log\left(\frac{n!}{\epsilon^n}\right)=O(n \log(n/\epsilon))\nonumber~.
\end{align}
\begin{equation}
\textrm{Extra time needed}\leq \sum_{t=0}^{T}\lvert S_{t+1}\rvert\log\left(\frac{\lvert S_{t+1}\rvert}{\epsilon}\right)\leq \sum_{t=1}^{n}t\log\left(\frac{t}{\epsilon}\right)=O(n^2 \log(n/\epsilon))\nonumber~.\tag*{\Halmos}
\end{equation}

\endproof

\subsection{On the Virtues of Our Bidder Optimal Core selection Rule}
\label{sec:discussion}
\noindent{\textbf{Water-filling versus cutting plane methods.}} As mentioned earlier in Section~\ref{sec:related-work}, given access to the separation oracle for the core polytope (which is essentially the procedure described in Definition~\ref{def:win-det}), one can use Vaidya's algorithm~\citep{vaidya1989new,vaidya1996new} or even the faster cutting plane methods for (feasibility) linear programmings such as \citealp{lee2015faster} to find a minimum revenue core point. As a recap, Vaidya's algorithm is deterministic and finds an  $\epsilon$-close minimum revenue core point with oracle complexity of $O\left((n/\epsilon) \log n\right)$ and extra computation of $O\left(n^4/\epsilon\right)$. The algorithm in \citealp{lee2015faster} is a randomized algorithm, and finds an $\epsilon$-close minimum revenue core point with high probability using $O\left(n^{(1+\epsilon)} \log^{O(1)}(n/\epsilon)\right)$ oracle evaluations in expectation and additional time $O\left(n^3\log^{O(1)}(n/\epsilon)\right)$. Relative to this (more general) approach, our water-filing algorithm (Algorithm~\ref{alg:bidderopt}) ({i})~requires asymptotically fewer oracle calls (i.e., $O(n\log(n/\epsilon)$ query calls) and less additional time (i.e., $O\left(n^2\log(n/\epsilon)\right)$ time); see Theorem~\ref{thm:main}, (ii) has an improved oracle complexity as a function of the precision parameter $\epsilon$, which hugely matters in the application domain we are interested in (i.e., the sale of ad space problem), (iii)~is deterministic, which is advantageous for the incentive and equilibrium properties of bidder optimal core points, (iv)~is a simple combinatorial algorithm which makes it more interpretable and easy-to-understand, and finally (v)~is easier to implement and  less sensitive to the choice of the precision parameter (and in general has less number of parameters to tune compared to the aforementioned cutting plane methods); see Section~\ref{sec:practical-considerations} in the supplementary material for a list of practical issues that we observed in our implementation of Vaidya's algorithm. These points highlight the advantages of our proposed core pricing algorithm, and in particular its practicality to be used in the sponsored search auction.

\noindent{\textbf{Water-filling versus VCG.}} VCG is a combinatorial auction in which reporting true values is a dominant strategy equilibrium. Computing payments in this auction can be done with $O(n)$ query calls to an optimal welfare oracle. Note also that by relaxing the incentive constraints to hold in expectation, approximate VCG payments can be computed implicitly by just one call to the allocation oracle in single dimensional~\citep{babaioff2010truthful}  or combinatorial~\citep{wilkens2015single} settings. In comparison, our water-filling algorithm (Algorithm~\ref{alg:bidderopt}) induces a truncated strategy (i.e., $i\in N: b_i(\cdot)=\max(v_i(\cdot)-\pi_i,0)$) that is a full-information $\epsilon$-Nash equilibrium, as its payment rule is an $\epsilon$-bidder optimal core payment~\citep{day2007fair}. Moreover, as we showed in Theorem~\ref{thm:main}, computing payments only requires $\tilde{O}(n)$ query calls to truncated maximum welfare allocation oracle. Furthermore, our payments produce more revenue at equilibrium~\citep{day2008core} and no coalition (subset of all bidders) can form a mutually beneficial renegotiation among themselves~\citep{day2008core} (to compare the revenue, one can consider the coalition $N\setminus\{i\}$ and note that  $\pi_i\leq w(N)-w(N\setminus\{i\})$, which is indeed the VCG utility of player $i$. Therefore, total generated revenue is lower-bounded by the VCG revenue).

\noindent{ \textbf{Parametrizing the path in water-filling.}} The choice of direction for water-filling at each iteration $t$ of Algorithm~\ref{alg:bidderopt} is flexible, as long as it is only restricted to increasing utilities for bidders in set $S_t$. Therefore it can potentially implement different core payment rules tailored for different economic objectives.  E.g., uniform water-filling will guarantee the approximate-equity of utilities subject to being a bidder optimal core point. Another desired objective is fairness with respect to VCG. For this objective, we can consider a variant of our algorithm that at each iteration performs \texttt{Core-Search} along the ray that connects $\pi^{(t)}$ to VCG. This ``VCG-pursuit" heuristic finds a bidder optimal point that attempts to minimize the angle with the ray connecting the origin with the VCG outcome, and therefore it heuristically implements an equilibrium in which winning bidders receive (almost) the same fraction of their utilities as in VCG. See Figure~\ref{fig:vcg-pursuit} for an example.

Notably,  our algorithm is not the only one in the literature that poses ``degrees of flexibility" to be adapted for different applications and economic goals. For example, \citealp{day2007fair} also propose a refinement of their algorithm when the set of bidder optimal payments is large. In that sense, their approach has a degree of flexibility. Another example is the computational search approach in \citealp{bunz2018designing,bunz2018designingEC} that provides a wide variety of core selecting rules based on different target goals (which can be fairness, revenue, or efficiency at an approximate Bayes Nash equilibrium). 

\begin{figure}[htb]
\begin{center}
 \includegraphics[scale=0.7]{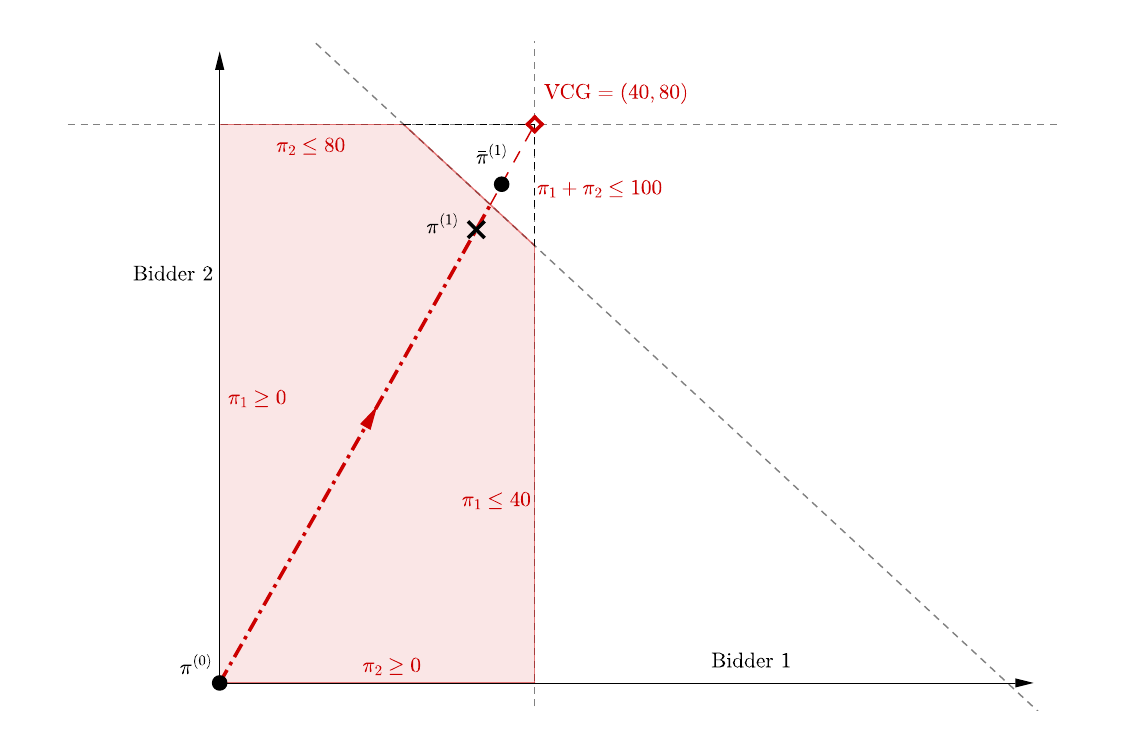}
\caption{VCG-Pursuit. In this variation of our water-filling algorithm at every iteration the algorithm tries to move along a ray that connects the current point to the VCG point (which can be outside of the core). }
\label{fig:vcg-pursuit}
\end{center}
\end{figure}

\section{Sale of Ad Space: a Numerical Study Based on Microsoft Bing Data}
\label{sec:experiment}
In this section, we experimentally study our proposed algorithm  (Algorithm~\ref{alg:bidderopt}) for sale of ad space by simulating it on ad auction bidding data. Our bidding data is collected through the Microsoft Bing search engine and its sponsored search auction platform. See Section~\ref{sec:bid-collection} in the supplementary material for more details on data collection and limitations on estimating true valuations. Using this dataset, we estimate and compare revenue, fairness  (this metric will be detailed later), running time, and number of calls to the (truncated) winner determination oracle for our algorithm and various benchmark auctions.


\subsection{Sale of Ad Space Allocation Problem}
\label{sec:ad-space-allocation}
We begin by describing the allocation problem and the bidding language for the sale of ad space problem, also known as rich ads auction (as in the title of this paper). 
In this problem, advertisers bid on search queries, called \emph{keywords}.  For each keyword that an advertiser bids on, she provides a collection of basic ads and, potentially, enhanced configurations (e.g., decorations that extend the size of an ad).  When a user searches for a keyword,  
a maximum amount of space to devote to ads (i.e., a specific number of lines) is exogenously determined, and an auction is used to determine which ads and in what order to show in that space.
The auction also specifies a per-click payment for each winning advertiser, denoted by \emph{cost per click (CPC)}, which will be charged to the advertiser if a user clicks on her winning ad in the auction.
The total expected revenue 
is then estimated by a winning ad's \emph{probability of click} ($p_{\text{click}}$) times the CPC of its corresponding advertiser, when summed over all of the winning ads.  Click rates are modeled and estimated by the search platform. Notably, this problem is different from the fully substitute separable model and suffers from \emph{complementarities} because of its knapsack nature (which are observable in our experiments, as we will see later ).

We use the Bing click prediction estimates in our experiment, which are based on the keyword, user, and few other contextual parameters. While we do not directly observe an advertiser's true value for an advertising assignment, we will estimate it as the probability of click, $p_{\text{click}}$, times the advertiser's bid amount $b$ on that ad.  In this sense, our experiments are 
essentially using the declared bid as a \emph{proxy} for true values (not assuming that the auction is truthful), which is a common practice in the sponsored search auction industry. See Section~\ref{sec:bid-collection} in the supplementary material for a comprehensive discussion.

The goal of winner determination allocation problem is to determine which ads to show in the available space (and which advertisers to pick as winners) to maximize the (declared) welfare.  This allocation is subject to various feasibility constraints.  In particular, we impose that 1) each advertiser can place at most one ad, 2) the total number of ads is bounded by a constant, and 3) the total number of lines must not exceed the available space.  In practice the platform may face other constraints as well, such as measures of user and advertiser quality; we omit those for simplicity.

\begin{definition}[Winner Determination Allocation Problem]
There is a slate of space, divided into $k$ lines (i.e, each is a single line of text).  A feasible allocation provides each advertiser with a consecutive block of lines.  Each advertiser $i=\{1,2,\ldots,n\}$ has a set of possible ads $\mathcal{D}_i$ with different lengths (which we call \emph{decorations}). She has a bid $b_i(d)$ and a probability of click $p_{i}(d)$ for each decoration $d\in\mathcal{D}_i$.   The total number of ads is $m\triangleq\sum_{i=1}^n\lvert\mathcal{D}_i\rvert$. The problem is then to select an ordered set of at most $h$ ads, one ad per advertiser, so that the total number of lines is no greater than $k$ and the total expected (declared) welfare of winning ads is maximized, where the expected (declared) value of a winning ad $(i,d)$ is equal to $p_i(d)\times b_i(d)$.
\end{definition}

For such advertiser preferences, we note that the welfare-maximizing allocation can be computed efficiently, i.e., in polynomial running time in the number of ads $m$ and the number of lines $k$, via dynamic programming (DP). This is also the case even when valuations are truncated as in Definition~\ref{def:win-det}, meaning that we give a different input to the algorithm where each advertiser has a truncation amount, and we essentially bring down the bids for all of her decorations by the that same truncation amount.  This  optimal algorithm is essentially a simple dynamic programming for the knapsack problem with two modifications (1) one additional dimension to limit the number of assigned items, (2) having different classes of items (one class for each advertiser), and being allowed to pick only one item from each class.  With only the latter condition, this problem is an instance of the classic multiple-choice knapsack problem~\citep{sinha1979multiple} and our DP is very similar to the solution for this problem. See Section~\ref{apx:dp} in the supplementary material (in particular Algorithm~\ref{alg:dp}) for details of our DP. See also Section~\ref{rem:second} in the supplementary material for a discussion on using faster approximation alternatives.

\subsection{Our experiments}
\label{sec:our-experiments}

We implemented the following list of auctions/algorithms, and simulated  them on the collected Bing bidding data in several experiments (the code for our implementation is provided as supplementary materials): Vickrey-Clarke-Groves auction (\texttt{VCG}), GSP with optimal welfare allocation (\texttt{GSP with Optimal}), GSP with greedy allocation (\texttt{GSP with Greedy}), minimum revenue core payment rule of \cite{day2007fair} (\texttt{Min Rev Core}), quadratic core payment rule of \cite{day2012quadratic} (\texttt{Quad Core}), Vaidya's cutting plane method (\texttt{Vaidya Min Rev}) for minimum revenue core, and our algorithm (\texttt{Fast Core}). See Section~\ref{sec:different-auctions} in the supplement for details of each auction/pricing rule.

Our numerical experiments are based on a dataset of $19,996$ auctions chosen uniformly at random from all of the Bing search advertising auctions that ran on January 23rd, 2018. The average number of different ads participating in each auction is $107$ and the average number of different advertisers is $7$. Therefore, the average number of configurations/decorations per advertiser is approximately $15$. The maximum number of ads that can be allocated in each auction is at most $4$. The maximum number of line counts that can be assigned, however, is variable across auction instances. We run a separate experiment on all of the auctions with the same line count. We consider line counts equal to 
$25$, $30$, $35$, $40$ and $45$. We report the corresponding results for each one of these cases. In all of these numerical experiments, we set $\epsilon=0.01$ (a parameter needed in Algorithm~\ref{alg:bidderopt}), and tried our best to optimize the parameters needed by other algorithms through trial-and-error for a fair comparison.

\paragraph{\textbf{Revenue.}}

Table~\ref{tab:revenue} shows the average expected revenue of the various algorithms we have studied in this paper, for different line counts. Figure~\ref{fig:revenue-40} shows how the revenue of each algorithm changes as a function of the total number of ads (or decorations) of all the advertisers attending the auction for a fixed line count=$40$. Due to Bing's policy of protecting advertisers data, we are only allowed to report the normalized revenue, i.e.,  normalized by the average VCG revenue for each line count in Table~\ref{tab:revenue}, and normalized by the average VCG revenue among all the auctions with the same line count of $40$ in Figure~\ref{fig:revenue-40}. Note that (1) we are \emph{not} normalizing by VCG revenue for each number of ads separately, and (2) All three of \texttt{Min Rev Core}, \texttt{Quad Core} and \texttt{Vaidya Min Rev} are minimum revenue core points (which are also bidder optimal), and hence obtain the same revenue. On the contrary, \texttt{Fast Core} only finds a bidder optimal point, which is not necessarily minimum revenue.

\begin{table}[htb]
\begin{center}
\footnotesize
\begin{tabular}{ |c||c|c|c|c|c|  }
	\hline
Line count & \texttt{VCG} & \texttt{GSP Optimal} & \texttt{GSP Greedy} &\shortstack{\\ \texttt{Minimum Revenue Core}\\ \citep{day2007fair};\\\citep{day2012quadratic};\\ \citep{vaidya1996new,vaidya1989new}}& \texttt{Fast Core}\\
	\hline
	\hline
25&	1.00&	1.252&	1.151 &1.149 &1.264\\
30&	1.00&	1.267&	1.239 &1.148 &1.265\\
35&	1.00&	1.282&	1.234 &1.150&1.265\\
40&	1.00&	1.283&	1.250 &1.154 &1.269\\
45&	1.00&	1.301&	1.287 &1.153 &1.269\\

	\hline
\end{tabular}

\end{center}
\vspace{2mm}
\caption{Average revenue per auction normalized by the revenue of VCG for each line count}	
\label{tab:revenue}
\end{table}

\begin{figure}[htb]
\hspace*{-2cm}
\includegraphics[width=6.5in]{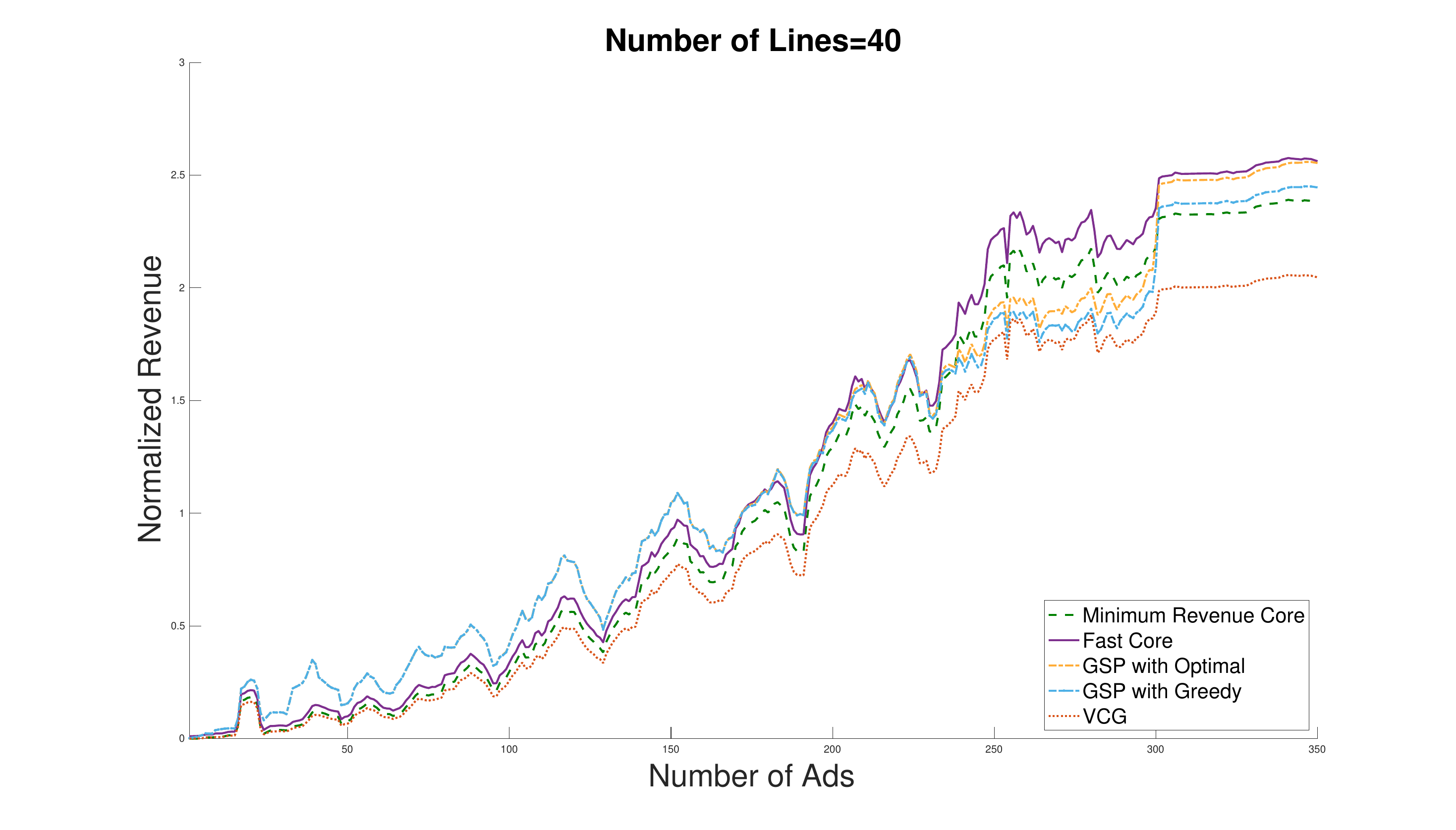}
\centering
\caption{Normalized revenues versus total number of ads for line count=40.\label{fig:revenue-40}}
\end{figure}

Table \ref{tab:revenue} shows that selecting a core outcome can significantly boost revenue compared to VCG. For different line counts, our core pricing algorithm attains $\approx$26\% improvement over VCG, and all other core pricing algorithm (which find a minimum revenue core point) obtain $\approx$15\% over VCG. Also, our core pricing algorithm obtains around 5\% improvement over \texttt{GSP Greedy}, and almost matches the average revenue of \texttt{GSP Optimal} (with less than 1\% loss on average). Notably, the minimum revenue core algorithms (i.e., \texttt{Min Rev Core}, \texttt{Quad Core} and \texttt{Vaidya Min Rev}) obtain less revenue on average than both versions of GSP for all different line counts. Finally, as seen in Figure~\ref{fig:revenue-40}, the ratio between the revenue of our core pricing algorithm and all other methods increases as the number of ads increases, e.g., our fast core pricing algorithm generates at least 6\% more revenue compared to \texttt{GSP Optimal} if number of ads is at least $300$ for line count $40$. See Section~\ref{appendix:experiments} in the supplement for similar figures for other line counts (Figures~\ref{fig:revenue-25}, \ref{fig:revenue-30}, \ref{fig:revenue-35}, and \ref{fig:revenue-45}).

\begin{remark} 
\label{remark:complement}
Having ads with different decorations creates complementarity in valuations, and our experimental results suggest that these complementarities are salient: there is a gap between the revenue of VCG and all of the core outcomes we study (which are all bidder Pareto optimal core points). Therefore, the shape of our core is not a hypercube, and VCG is not in the core (see \citealp{day2008core} for a proof). Also, our experiments show that the bidder Pareto optimal frontier of the core contains points that are \emph{not} necessarily minimum revenue (as there is a revenue gap between minimum revenue core points found by different methods and our bidder optimal core point). 
\end{remark}
For completeness, consider the following example that theoretically proves VCG is not necessarily in the core in the sale of ad space problem, which verifies the experimental observation in Remark~\ref{remark:complement}.
\begin{example} 
\label{example:complement}
Consider an instance of the sale of ad space problem with $k=9$ lines and $h=2$ slots. Suppose all the ads have $p_\textrm{click}=0.5$. There are $5$ advertisers. $A_1$ has $2$ decorations: one with bid $10$ and number of lines $3$, and one with bid $20$ and number of lines $6$. $A_2$ has one decoration with bid $31$ and number of lines $8$. $A_3$ has one decoration with bid $15$ and number of lines $5$. $A_4$ has one decoration with bid $11$ and number of lines $3$. Finally, $A_5$ has one decoration with bid $17$ and number of lines $4$. In the optimal welfare allocation, bidders $A_3$ and $A_5$ win (with total generated welfare of $16$). Moreover, VCG prices of these bidders (after multiplying with $p_\textrm{click}$) are $7$ and $8$ respectively. Interestingly, this point is not in the core, as simple calculations show that the core polytope is the set of all the prices satisfying $p_3+p_5\geq 15.5$, $p_3\in[7,15]$, $p_5\in[8,15.5]$.
\end{example}
\begin{remark}
As mentioned earlier, our experiments use declared values as proxies for true valuations. Under this treatment, our revenue report should basically be interpreted as total payments of our bidder optimal core point in the core with respect to true valuations. Note that core auctions are not truthful, but bidder optimal core points (similar to those in this paper) admit a natural full information Nash equilibrium, so that revenue under this outcome is equal to total payments in the core with respect to true valuations~\citep{day2008core}. Notably, the non-truthfulness of core auction impose these limitations in our experiments. We provide more details on the impact of incentives on short/long term revenue, other interpretations of our results, and future avenues for research to overcome these limitations in Section~\ref{sec:incentive-core} in the supplement.
\end{remark}
\paragraph{\textbf{Running time.}}
Table \ref{tab:time} shows the average normalized running times of different auctions for various line counts, where we normalized with respect to the average running time of VCG. We also report the average running time of VCG in milliseconds in Table \ref{tab:time}. Figure~\ref{fig:runtime-40} shows how normalized running times change as the number of ads increases for the line count $40$. Here, the normalization is with respect to the average running time of VCG for each specific number of ads. See Section~\ref{appendix:experiments} in the supplement for similar figures for other line counts (Figures~\ref{fig:runtime-25}, \ref{fig:runtime-30}, \ref{fig:runtime-35}, and \ref{fig:runtime-45}).

\begin{table}[htb]
\footnotesize
\begin{center}
\hspace*{-1cm}
\begin{tabular}{ |c||c|c|c|c|c|c|c|  }
	\hline
Line count & \shortstack{\\ \texttt{VCG}\\ running time} & \shortstack{\\ \texttt{GSP}\\\texttt{Optimal}}& \shortstack{\\ \texttt{GSP}\\ \texttt{Greedy} }&\shortstack{\\ \texttt{Mini Rev Core}\\ \citep{day2007fair}}&\shortstack{\\ \texttt{Quad Core}\\ \citep{day2012quadratic}}&\shortstack{\\ \texttt{Vaidya Min Rev}\\ \citep{vaidya1996new,vaidya1989new}}& \shortstack{\\ \texttt{Fast}\\\texttt{Core}}\\
	\hline
	\hline
25&	 5.264 (ms)&	0.50&	0.05 &47.25 &86.47 &59.00 &7.53\\
30&	6.552 (ms)&	0.49&	0.05 &44.15 &77.15 &53.19&7.44\\
35&	7.465 (ms)&	0.49&	0.05 &44.67&75.16&49.25 &7.45\\
40&	8.648 (ms)&	0.49&	0.05 &42.22 &68.99&50.19 &7.42\\
45&	10.163 (ms)&	0.48&	0.05 &40.13 &64.49&47.66 &7.39\\

	\hline
\end{tabular}
\end{center}
\vspace{2mm}
\caption{Average running time per auction normalized by the average running time of VCG for each line count}	
\label{tab:time}
\end{table}

\begin{figure}[htb]
\hspace*{-2cm}
\includegraphics[width=6.5in]{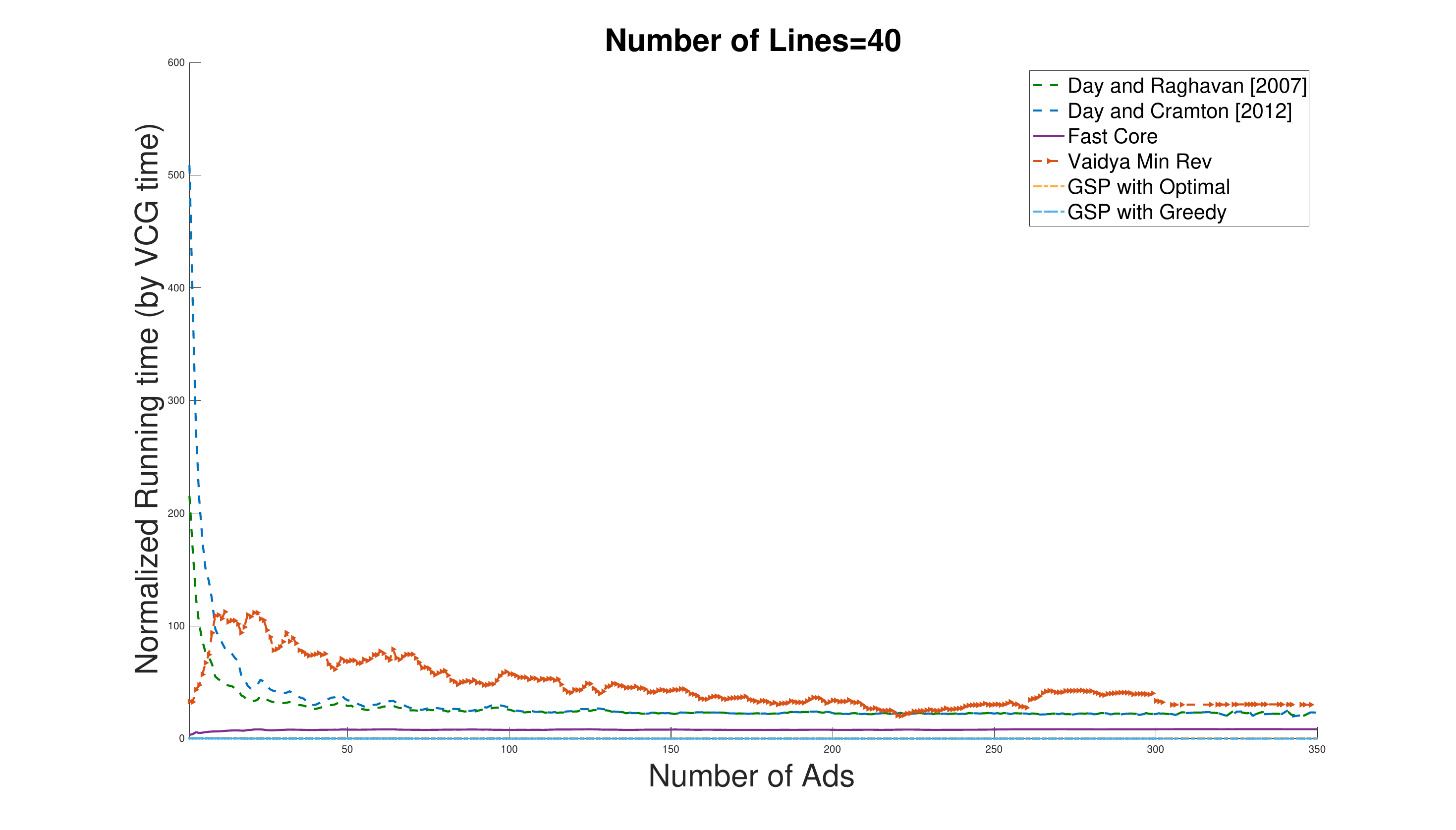}
\centering
\caption{Normalized running times versus total number of ads for line count=40.\label{fig:runtime-40}}
\end{figure}

As it can be seen from both Table~\ref{tab:time} and Figure~\ref{fig:runtime-40}, the running time of our core pricing algorithm is no more than $\approx7.5$ times the VCG running time, which is still in an acceptable range for sponsored search auction of Bing. As expected, both variants of GSP have much smaller running time compared to all other auctions (as \texttt{GSP Optimal} only calls the winner determination oracle once, and \texttt{GSP Greedy} does not even need to call it once). The running time of all other core pricing algorithms is at least $6-10$ times the running time of our core pricing algorithm. In particular, as seen in Figure~\ref{fig:runtime-40}, all three of \citealp{day2007fair}, \citealp{day2012quadratic} and Vaidya's algorithm have larger running time when the number of ads is small (smaller than $10$), and they obtain a more comparable running time as the number of ads increase (still, slower than our core pricing algorithm by a factor of $5-8$). Still, our experiments never identified a scenario in which Vaidya or any other core pricing rule could beat our algorithm in terms of running time.

\begin{remark} It should be noted that the convergence of LP-based (or quadratic programming based) heuristics, as well as interior point methods such as Vaidya to solve the minimum revenue core LP, is highly sensitive to selected internal parameters of these algorithm (e.g., the initialization point, the termination condition, or the barrier thresholds in Vaidya), as our experiments suggest. On the contrary, our core pricing algorithm only needs an assignment for $\epsilon>0$, and our experiments indicate that its performance is not very much sensitive to the choice of this parameter.
\end{remark}
\begin{remark} We observed somewhat frequent unstable behavior when running Vaidya's algorithm in our application, e.g., in some auctions by slightly changing the barrier thresholds or the initial point the algorithm had a very slow convergence. See Section~\ref{sec:practical-considerations} in the supplement for a comprehensive discussion.
\end{remark}

\paragraph{\textbf{Query complexity.}} To have a more comprehensive comparison between the speed of different algorithms, and to compare how efficiently they utilize the winner determination oracle, we also compare the query complexity of various core pricing algorithms. Table~\ref{tab:query} shows the average query complexity, i.e., number of calls to the oracle, of different core pricing algorithms and how they are compared with the VCG auction as a benchmark (which is basically equal to the average number of advertisers winning the auction). Figure~\ref{fig:query-40} shows how the query complexity of different core pricing algorithms vary as the number of ads increases for the line count $40$. See Section~\ref{appendix:experiments} in the supplement for similar figures for other line counts (Figures~\ref{fig:query-25}, \ref{fig:query-30}, \ref{fig:query-35}, and \ref{fig:query-45}). 

\begin{table}[htb]
\footnotesize
\begin{center}
\hspace*{-1cm}
\begin{tabular}{ |c||c|c|c|c|c|  }
	\hline
Line count & \texttt{VCG}&\shortstack{\\ \texttt{Mini Rev Core}\\ \citep{day2007fair}}&\shortstack{\\ \texttt{Quad Core}\\ \citep{day2012quadratic}}&\shortstack{\\ \texttt{Vaidya Min Rev}\\ \citep{vaidya1996new,vaidya1989new}}& \shortstack{\\ \texttt{Fast}\\\texttt{Core}}\\
	\hline
	\hline
25&	 3.740&	19.750&	20.513 &75.584 &11.648 \\
30&	3.757&23.824&	24.574 &69.478 &11.741 \\
35&	3.763&23.797&	24.572&64.393&11.741\\
40&	3.764&23.669&	24.439 &65.944 &11.754\\
45&	3.765&	23.639&	24.400 &62.867 &11.753\\

	\hline
\end{tabular}
\end{center}
\vspace{2mm}
\caption{Average query complexity for each line count}	
\label{tab:query}
\end{table}

\begin{figure}[htb]
\hspace*{-2cm}
\includegraphics[width=6.5in]{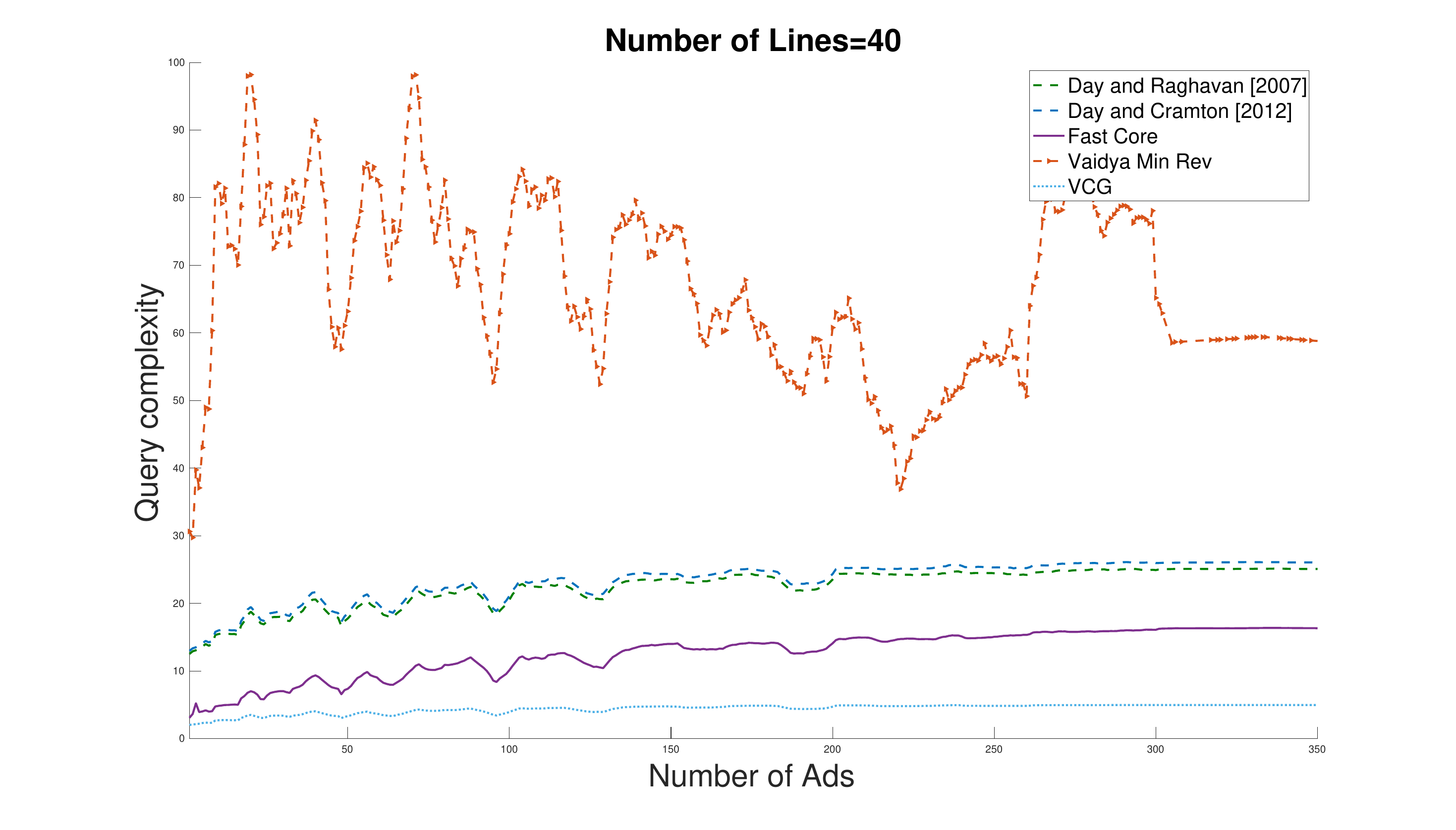}
\centering
\caption{Query complexity versus total number of ads for line count=40.\label{fig:query-40}}
\end{figure}
As it is clear from Table~\ref{tab:query}, our core pricing algorithm makes drastically smaller number of calls to the oracle compared to other core pricing methods: core pricing rules of \citealp{day2007fair} and \citealp{day2012quadratic} make $\approx 2.1$ times, and Vaidya's algorithm makes $\approx 5.9$ times number of calls on average to the winner determination oracle compared to Algorithm~\ref{alg:bidderopt}. Moreover, our core pricing algorithm makes $\approx 3.1$ times more number of calls than the VCG baseline (which still is an acceptable range for our application). Finally, as can be seen in Figure~\ref{fig:query-40}, Vaidya's algorithm has a rather high variance in its number of query calls compared to other methods, which is another factor that makes it unstable to be used for our application. See Section~\ref{sec:practical-considerations} in the supplement for more details.
\paragraph{\textbf{Fairness.}}Another measure we study is the fairness to the advertisers. For a given auction, our notion of fairness is the ratio of the maximum utility to the minimum utility  among the winning advertisers, i.e., $\frac{\max_{\textrm{ad} \in \mathcal{A}} p(\textrm{ad})\left(b(\textrm{ad})- \text{CPC}(\textrm{ad})\right)}{\min_{\textrm{ad} \in \mathcal{A}} p(\textrm{ad})\left(b(\textrm{ad}) - \text{CPC}(\textrm{ad})\right)}$, where $\mathcal{A}$ is the set of all of the winning ads. Figure~\ref{fig:fairness-40} shows how the fairness ratio of different algorithms change as the number of ads increases for the line count $40$. See Section~\ref{appendix:experiments} in the supplement for other line counts (Figures~\ref{fig:fairness-25}, \ref{fig:fairness-30}, \ref{fig:fairness-35}, and \ref{fig:fairness-45}). 

\begin{figure}[htb]
\hspace*{-2cm}
\includegraphics[width=6.5in]{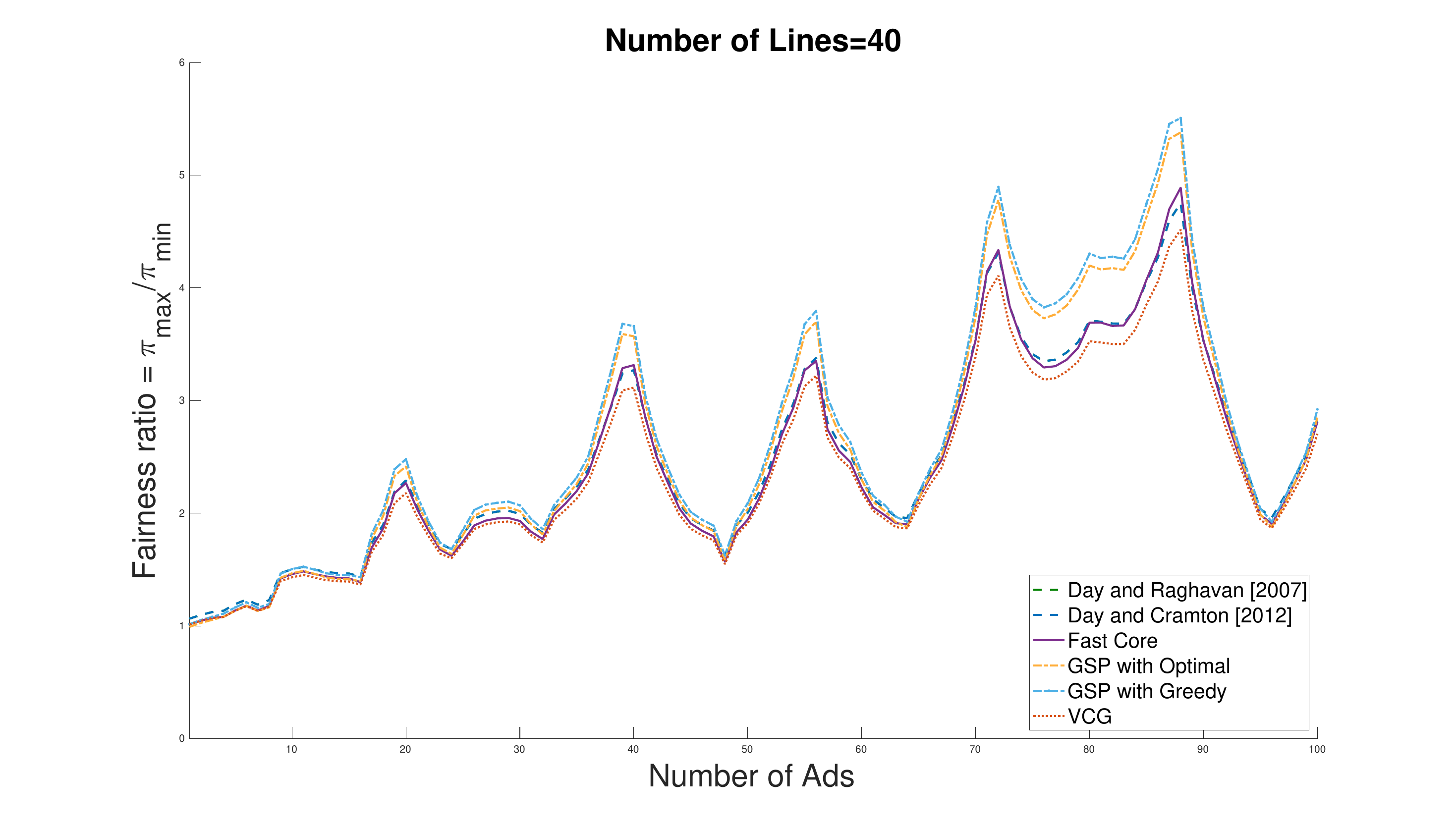}
\centering
\caption{Fairness ratio versus total number of ads for line count=40.\label{fig:fairness-40}}
\end{figure}

 As can be seen in Figure~\ref{fig:fairness-40}, VCG has the maximum fairness (i.e., the fairness ratio is closer to $1$) among other methods. All of the core selecting rules (including ours) have better fairness ratios compared to the two variants of GSP. One interpretation is that when we uniformly increase utilities of  advertisers in Algorithm~\ref{alg:bidderopt}, we tend to fairly divide the extra surplus between eligible advertisers while not violating the core constraints. Note that a core constraint, however, might cause a specific advertiser to pay a large amount compared to the VCG prices, and hence caps the utility of this advertiser. Consequently,  the core tends to lie between VCG and GSP in terms of fairness.

\section{ Conclusion}

We proposed a fast algorithm, with provable correctness and running time, that finds a bidder optimal core point with almost linear calls to the winner determination oracle for truncated valuations. Our running time advantage compared to prior work is both in terms of query complexity of calling this oracle and extra computations needed. We demonstrated  our algorithm in the time-sensitive sale of the ad space application through  numerical experiments on Microsoft Bing auction data. Our results suggested a considerable revenue improvement over VCG by our core pricing rule, but only with respect to reported valuations. We also studied the fairness of our core pricing rule compared to other methods using a simple measure, and our experimental study showed our algorithm retains acceptable fairness properties.  See Section~\ref{apx:openproblems} in the supplement for future directions and open problems.



\section*{Acknowledgment}
The authors would like to thank Bob Day, Larry Ausubel, Paul Milgrom, Sven Seuken, and Microsoft Bing Ad Auction team for their insightful comments. We would also like to thank the anonymous referees and the associate editor for extraordinarily helpful comments during the revision process. The second author was supported in part by NSF CCF 1618502.

%
%
%





\setlength{\bibsep}{0.0pt}
\bibliographystyle{plainnat}
\OneAndAHalfSpacedXI
{\footnotesize
\bibliography{refs}}

\begin{thebibliography}{50}
\providecommand{\natexlab}[1]{#1}
\providecommand{\url}[1]{\texttt{#1}}
\expandafter\ifx\csname urlstyle\endcsname\relax
  \providecommand{\doi}[1]{doi: #1}\else
  \providecommand{\doi}{doi: \begingroup \urlstyle{rm}\Url}\fi

\bibitem[Aggarwal et~al.(2006)Aggarwal, Goel, and
  Motwani]{aggarwal2006truthful}
Gagan Aggarwal, Ashish Goel, and Rajeev Motwani.
\newblock Truthful auctions for pricing search keywords.
\newblock In \emph{Proceedings of the 7th ACM conference on Electronic
  commerce}, pages 1--7. ACM, 2006.

\bibitem[Aggarwal et~al.(2019)Aggarwal, Badanidiyuru, and
  Mehta]{aggarwal2019autobidding}
Gagan Aggarwal, Ashwinkumar Badanidiyuru, and Aranyak Mehta.
\newblock Autobidding with constraints.
\newblock In \emph{International Conference on Web and Internet Economics},
  pages 17--30. Springer, 2019.

\bibitem[Ausubel and Baranov(2010)]{ausubel2010core}
Lawrence Ausubel and Oleg~V Baranov.
\newblock Core-selecting auctions with incomplete information.
\newblock \emph{University of Maryland}, 121, 2010.

\bibitem[Ausubel and Cramton(1999)]{ausubel1999optimality}
Lawrence Ausubel and Peter Cramton.
\newblock The optimality of being efficient.
\newblock \emph{University of Maryland}, 1999.

\bibitem[Ausubel and Milgrom(2002)]{ausubel2002ascending}
Lawrence Ausubel and Paul Milgrom.
\newblock Ascending auctions with package bidding.
\newblock \emph{Advances in Theoretical Economics}, 1\penalty0 (1), 2002.

\bibitem[Ausubel and Milgrom(2006)]{ausubel2006lovely}
Lawrence Ausubel and Paul Milgrom.
\newblock The lovely but lonely vickrey auction.
\newblock \emph{Combinatorial auctions}, 17:\penalty0 22--26, 2006.

\bibitem[Babaioff et~al.(2015)Babaioff, Kleinberg, and
  Slivkins]{babaioff2010truthful}
Moshe Babaioff, Robert~D Kleinberg, and Aleksandrs Slivkins.
\newblock Truthful mechanisms with implicit payment computation.
\newblock \emph{Journal of the ACM (JACM)}, 62\penalty0 (2):\penalty0 1--37,
  2015.

\bibitem[Beck and Ott(2009)]{beck2009revenue}
Marissa Beck and Marion Ott.
\newblock Revenue monotonicity in core-selecting package auctions.
\newblock \emph{Working paper}, 2009.

\bibitem[Bosshard et~al.(2017)Bosshard, B{\"u}nz, Lubin, and
  Seuken]{bosshard2017computing}
Vitor Bosshard, Benedikt B{\"u}nz, Benjamin Lubin, and Sven Seuken.
\newblock Computing bayes-nash equilibria in combinatorial auctions with
  continuous value and action spaces.
\newblock In \emph{IJCAI}, pages 119--127, 2017.

\bibitem[Bubeck et~al.(2015)]{bubeck2015convex}
S{\'e}bastien Bubeck et~al.
\newblock Convex optimization: Algorithms and complexity.
\newblock \emph{Foundations and Trends{\textregistered} in Machine Learning},
  8\penalty0 (3-4):\penalty0 231--357, 2015.

\bibitem[B{\"u}nz et~al.(2015)B{\"u}nz, Seuken, and Lubin]{bunz2015faster}
Benedikt B{\"u}nz, Sven Seuken, and Benjamin Lubin.
\newblock A faster core constraint generation algorithm for combinatorial
  auctions.
\newblock In \emph{Proceedings of the Twenty-Ninth AAAI Conference on
  Artificial Intelligence}, pages 827--834. AAAI Press, 2015.

\bibitem[B{\"u}nz et~al.(2018{\natexlab{a}})B{\"u}nz, Lubin, and
  Seuken]{bunz2018designing}
Benedikt B{\"u}nz, Benjamin Lubin, and Sven Seuken.
\newblock Designing core-selecting payment rules: A computational search
  approach.
\newblock \emph{Available at SSRN 3178454}, 2018{\natexlab{a}}.

\bibitem[B{\"u}nz et~al.(2018{\natexlab{b}})B{\"u}nz, Lubin, and
  Seuken]{bunz2018designingEC}
Benedikt B{\"u}nz, Benjamin Lubin, and Sven Seuken.
\newblock Designing core-selecting payment rules: A computational search
  approach.
\newblock In \emph{Proceedings of the 2018 ACM Conference on Economics and
  Computation}, pages 109--109. ACM, 2018{\natexlab{b}}.

\bibitem[Cavallo et~al.(2017)Cavallo, Krishnamurthy, Sviridenko, and
  Wilkens]{cavallo2017sponsored}
Ruggiero Cavallo, Prabhakar Krishnamurthy, Maxim Sviridenko, and Christopher~A
  Wilkens.
\newblock Sponsored search auctions with rich ads.
\newblock In \emph{Proceedings of the 26th International Conference on World
  Wide Web}, pages 43--51. International World Wide Web Conferences Steering
  Committee, 2017.

\bibitem[Clarke(1971)]{clarke1971multipart}
Edward~H Clarke.
\newblock Multipart pricing of public goods.
\newblock \emph{Public choice}, 11\penalty0 (1):\penalty0 17--33, 1971.

\bibitem[Cramton(2013)]{cramton2013spectrum}
Peter Cramton.
\newblock Spectrum auction design.
\newblock \emph{Review of Industrial Organization}, 42\penalty0 (2):\penalty0
  161--190, 2013.

\bibitem[Day and Cramton(2012)]{day2012quadratic}
Robert Day and Peter Cramton.
\newblock Quadratic core-selecting payment rules for combinatorial auctions.
\newblock \emph{Operations Research}, 60\penalty0 (3):\penalty0 588--603, 2012.

\bibitem[Day and Milgrom(2008)]{day2008core}
Robert Day and Paul Milgrom.
\newblock Core-selecting package auctions.
\newblock \emph{international Journal of game Theory}, 36\penalty0
  (3-4):\penalty0 393--407, 2008.

\bibitem[Day and Raghavan(2007)]{day2007fair}
Robert Day and Subramanian Raghavan.
\newblock Fair payments for efficient allocations in public sector
  combinatorial auctions.
\newblock \emph{Management Science}, 53\penalty0 (9):\penalty0 1389--1406,
  2007.

\bibitem[Dobzinski and Nisan(2007)]{dobzinski2007mechanisms}
Shahar Dobzinski and Noam Nisan.
\newblock Mechanisms for multi-unit auctions.
\newblock In \emph{Proceedings of the 8th ACM conference on Electronic
  commerce}, pages 346--351. ACM, 2007.

\bibitem[Edelman and Ostrovsky(2007)]{edelman2007strategic}
Benjamin Edelman and Michael Ostrovsky.
\newblock Strategic bidder behavior in sponsored search auctions.
\newblock \emph{Decision support systems}, 43\penalty0 (1):\penalty0 192--198,
  2007.

\bibitem[Edelman et~al.(2007)Edelman, Ostrovsky, and
  Schwarz]{edelman2007internet}
Benjamin Edelman, Michael Ostrovsky, and Michael Schwarz.
\newblock Internet advertising and the generalized second-price auction:
  Selling billions of dollars worth of keywords.
\newblock \emph{American economic review}, 97\penalty0 (1):\penalty0 242--259,
  2007.

\bibitem[Erdil and Klemperer(2010)]{erdil2010new}
Aytek Erdil and Paul Klemperer.
\newblock A new payment rule for core-selecting package auctions.
\newblock \emph{Journal of the European Economic Association}, 8\penalty0
  (2-3):\penalty0 537--547, 2010.

\bibitem[Goel and Khani(2014)]{goel2014revenue}
Gagan Goel and Mohammad~Reza Khani.
\newblock Revenue monotone mechanisms for online advertising.
\newblock In \emph{Proceedings of the 23rd international conference on World
  wide web}, pages 723--734. ACM, 2014.

\bibitem[Goel et~al.(2015)Goel, Khani, and Leme]{goel2015core}
Gagan Goel, Mohammad~Reza Khani, and Renato~Paes Leme.
\newblock Core-competitive auctions.
\newblock In \emph{Proceedings of the Sixteenth ACM Conference on Economics and
  Computation}, pages 149--166. ACM, 2015.

\bibitem[Goeree and Lien(2016)]{goeree2016impossibility}
Jacob~K Goeree and Yuanchuan Lien.
\newblock On the impossibility of core-selecting auctions.
\newblock \emph{Theoretical economics}, 11\penalty0 (1):\penalty0 41--52, 2016.

\bibitem[Gould and Toint(2004)]{gould2004preprocessing}
Nick Gould and Philippe~L Toint.
\newblock Preprocessing for quadratic programming.
\newblock \emph{Mathematical Programming}, 100\penalty0 (1):\penalty0 95--132,
  2004.

\bibitem[Gr{\"o}tschel et~al.(1981)Gr{\"o}tschel, Lov{\'a}sz, and
  Schrijver]{grotschel1981ellipsoid}
Martin Gr{\"o}tschel, L{\'a}szl{\'o} Lov{\'a}sz, and Alexander Schrijver.
\newblock The ellipsoid method and its consequences in combinatorial
  optimization.
\newblock \emph{Combinatorica}, 1\penalty0 (2):\penalty0 169--197, 1981.

\bibitem[Groves(1973)]{groves1973incentives}
Theodore Groves.
\newblock Incentives in teams.
\newblock \emph{Econometrica: Journal of the Econometric Society}, pages
  617--631, 1973.

\bibitem[Lamy(2010)]{lamy2010core}
Laurent Lamy.
\newblock Core-selecting package auctions: a comment on revenue-monotonicity.
\newblock \emph{International Journal of Game Theory}, 39\penalty0
  (3):\penalty0 503--510, 2010.

\bibitem[Lee et~al.(2015)Lee, Sidford, and Wong]{lee2015faster}
Yin~Tat Lee, Aaron Sidford, and Sam Chiu-wai Wong.
\newblock A faster cutting plane method and its implications for combinatorial
  and convex optimization.
\newblock In \emph{2015 IEEE 56th Annual Symposium on Foundations of Computer
  Science}, pages 1049--1065. IEEE, 2015.

\bibitem[Lubin and Parkes(2009)]{lubin2009quantifying}
Benjamin Lubin and David~C Parkes.
\newblock Quantifying the strategyproofness of mechanisms via metrics on payoff
  distributions.
\newblock In \emph{Proceedings of the Twenty-Fifth Conference on Uncertainty in
  Artificial Intelligence}, pages 349--358. AUAI Press, 2009.

\bibitem[Lubin et~al.(2015)Lubin, B\'{y}nz, and Seuken]{lubin2015new}
Benjamin Lubin, Benedikt B\'{y}nz, and Sven Seuken.
\newblock New core-selecting payment rules with better fairness and incentive
  properties.
\newblock In \emph{3rd Conference on Auctions, Market Mechanisms and Their
  Applications}. ACM, 2015.

\bibitem[Mehrotra(1992)]{mehrotra1992implementation}
Sanjay Mehrotra.
\newblock On the implementation of a primal-dual interior point method.
\newblock \emph{SIAM Journal on optimization}, 2\penalty0 (4):\penalty0
  575--601, 1992.

\bibitem[Milgrom(2007)]{milgrom2007package}
Paul Milgrom.
\newblock Package auctions and exchanges.
\newblock \emph{Econometrica}, 75\penalty0 (4):\penalty0 935--965, 2007.

\bibitem[Nisan and Ronen(2007)]{nisan2007computationally}
Noam Nisan and Amir Ronen.
\newblock Computationally feasible vcg mechanisms.
\newblock \emph{Journal of Artificial Intelligence Research}, 29:\penalty0
  19--47, 2007.

\bibitem[Osborne and Rubinstein(1994)]{osborne1994course}
Martin~J Osborne and Ariel Rubinstein.
\newblock \emph{A course in game theory}.
\newblock MIT press, 1994.

\bibitem[Rastegari et~al.(2011)Rastegari, Condon, and
  Leyton-Brown]{rastegari2011revenue}
Baharak Rastegari, Anne Condon, and Kevin Leyton-Brown.
\newblock Revenue monotonicity in deterministic, dominant-strategy
  combinatorial auctions.
\newblock \emph{Artificial Intelligence}, 175\penalty0 (2):\penalty0 441--456,
  2011.

\bibitem[Sandholm(2002)]{sandholm2002algorithm}
Tuomas Sandholm.
\newblock Algorithm for optimal winner determination in combinatorial auctions.
\newblock \emph{Artificial intelligence}, 135\penalty0 (1-2):\penalty0 1--54,
  2002.

\bibitem[Sano(2012)]{sano2012non}
Ryuji Sano.
\newblock Non-bidding equilibrium in an ascending core-selecting auction.
\newblock \emph{Games and Economic Behavior}, 74\penalty0 (2):\penalty0
  637--650, 2012.

\bibitem[Sayedi(2018)]{sayedi2018real}
Amin Sayedi.
\newblock Real-time bidding in online display advertising.
\newblock \emph{Marketing Science}, 37\penalty0 (4):\penalty0 553--568, 2018.

\bibitem[Sinha and Zoltners(1979)]{sinha1979multiple}
Prabhakant Sinha and Andris~A Zoltners.
\newblock The multiple-choice knapsack problem.
\newblock \emph{Operations Research}, 27\penalty0 (3):\penalty0 503--515, 1979.

\bibitem[Vaidya(1989)]{vaidya1989new}
Pravin~M Vaidya.
\newblock A new algorithm for minimizing convex functions over convex sets.
\newblock In \emph{30th Annual Symposium on Foundations of Computer Science},
  pages 338--343. IEEE, 1989.

\bibitem[Vaidya(1996)]{vaidya1996new}
Pravin~M Vaidya.
\newblock A new algorithm for minimizing convex functions over convex sets.
\newblock \emph{Mathematical programming}, 73\penalty0 (3):\penalty0 291--341,
  1996.

\bibitem[Vickrey(1961)]{vickrey1961counterspeculation}
William Vickrey.
\newblock Counterspeculation, auctions, and competitive sealed tenders.
\newblock \emph{The Journal of finance}, 16\penalty0 (1):\penalty0 8--37, 1961.

\bibitem[Wilkens and Sivan(2015)]{wilkens2015single}
Christopher~A Wilkens and Balasubramanian Sivan.
\newblock Single-call mechanisms.
\newblock \emph{ACM Transactions on Economics and Computation}, 3\penalty0
  (2):\penalty0 10, 2015.

\bibitem[Wilkens et~al.(2016)Wilkens, Cavallo, Niazadeh, and
  Taggart]{wilkens2016mechanism}
Christopher~A Wilkens, Ruggiero Cavallo, Rad Niazadeh, and Samuel Taggart.
\newblock Mechanism design for value maximizers.
\newblock \emph{arXiv preprint arXiv:1607.04362}, 2016.

\bibitem[Wilkens et~al.(2017)Wilkens, Cavallo, and Niazadeh]{wilkens2017gsp}
Christopher~A Wilkens, Ruggiero Cavallo, and Rad Niazadeh.
\newblock Gsp: the cinderella of mechanism design.
\newblock In \emph{Proceedings of the 26th International Conference on World
  Wide Web}, pages 25--32, 2017.

\bibitem[Xu et~al.(2013)Xu, Gao, Yang, and Liu]{xu2013predicting}
Haifeng Xu, Bin Gao, Diyi Yang, and Tie-Yan Liu.
\newblock Predicting advertiser bidding behaviors in sponsored search by
  rationality modeling.
\newblock In \emph{Proceedings of the 22nd international conference on World
  Wide Web}, pages 1433--1444. ACM, 2013.

\bibitem[Zhang(1998)]{zhang1998solving}
Yin Zhang.
\newblock Solving large-scale linear programs by interior-point methods under
  the matlab environment.
\newblock \emph{Optimization Methods and Software}, 10\penalty0 (1):\penalty0
  1--31, 1998.

\end{thebibliography}

\clearpage
 \begin{APPENDICES}
\section{Dynamic programming for sale of ad space}
\label{apx:dp}
Our simple dynamic programming is explained in Algorithm~\ref{alg:dp}. We assume (without loss) that the $m$ ads are indexed so that all ads from the same advertiser lie in a contiguous range.  That is, no ad of one advertiser comes between two ads from another advertiser.  Then the subproblem for our dynamic program is $\texttt{sub}(i, h', k')$, which gives the optimal welfare while using ads from index $i$ to $m$, allocating at most $h'$ ads from distinct advertisers and using at most $k'$ lines. The solution to the full allocation problem is then $\texttt{sub}(1, h, k)$.  We use $l\left(\text{ad}_i\right)$ as a notation for the number of lines for $\text{ad}_i$ and $W\left(\text{ad}_i\right) $ for the expected (declared) value of $\text{ad}_i$ to its associated advertiser. 

\begin{algorithm}
\caption{(Dynamic program for winner determination)}
\label{alg:dp}

\KwIn{ ads $\{\text{ad}_i\}_{i\in m}$, index $I$, number of ads to be allocated $h'$ and maximum number of available lines $k'$}
\vspace{1mm}

\textbf{initialize} $\text{bestWF} \leftarrow 0$. 

\uIf{$I > m$ or $h' < 1$ or $k' < 1$}{

\textbf{return} 0
}\Else{
\For{$i = I$ \textbf{to} $m$}
{Let $j$ be next ad from different advertiser than $i$.

\If{$\text{bestWF} < W\left(\text{ad}_i\right) + \text{sub}\left(j, h'-1, k' - l\left(\text{ad}_i\right)\right)$}
{
 $\text{bestWF} \leftarrow W\left(\text{ad}_i\right) + \text{sub}\left(j, h'-1, k' - l\left(\text{ad}_i\right)\right)$
}
Let $\tilde{\pi}_i=\pi_i+\Delta$ for $i\in S$, $\tilde{\pi}_i=\pi_i$ for $i\in N\setminus S$, and $\tilde{\pi}_0=w(N)-\sum_{i\in N}\tilde{\pi}_i$.
}
 \textbf{return} $\text{bestWF}$
}

\end{algorithm}


\section{Benchmark core pricing algorithms and other auctions}
\label{sec:different-auctions}
Here is a detailed list of payment rules/auctions we implemented in this paper: 
\renewcommand{\labelenumi}{\textbf{\roman{enumi}}.}
\begin{enumerate}
\item \underline{The Vickrey-Clarke-Groves auction (\texttt{VCG})}: See \citealp{vickrey1961counterspeculation};\citealp{clarke1971multipart};\citealp{groves1973incentives}.
\item \underline{The Generalized Second Price auction with optimal welfare allocation (\texttt{GSP with Optimal})}: this auction uses the optimal welfare allocation. For payments, similar to traditional GSP, it prices each ad according to the $p_{\text{click}}$ times bid of the subsequent ad (the last ad will be priced by the best ad that was not assigned and can be fit within the line count limit). 
\item \underline{The Generalized Second Price auction with greedy allocation (\texttt{GSP with Greedy})}: for the allocation, it greedily allocates ads based on $p_{\text{click}}$ times bids. For payments, similar to GSP, uses the next best ad for pricing. Greedy GSP has faster runtime, since it does need not to call the winner determination oracle, but has worse revenue performance relative to Optimal GSP (as we will see in our experimental results).
\item \underline{ Minimum revenue core payment rule (\texttt{Min Rev Core}; \citealp{day2007fair})}: this is the first and simplest heuristic algorithm that finds a minimum revenue core point, given access to an oracle for the winner determination problem (with truncated values). This algorithm is based on a heuristic called Core Constraint Generation (CCG). The simple version in \citealp{day2007fair} starts from an initial small LP  for minimizing revenue (which is basically the hyper-cube when payments are above VCG prices and below bids), and in each iteration finds the most violated core constraint by the current point (by sending a query to the winner determination oracle of Definition~\ref{def:win-det}) and adds this constraint to the LP. It then re-solves the LP to find the next point, and iterates until it finds a feasible core point. We use Matlab's large-scale LP solver based on  the interior-point methods~\citep{zhang1998solving} for the LP-solving part of this algorithm.
\item \underline{Quadratic core payment rule (\texttt{Quad Core}; \citealp{day2012quadratic})}: this algorithm finds the closest minimum revenue core point to VCG prices. It is again a heuristic algorithm that first uses CCG and finds the minimum revenue core point at each iteration similar to \citealp{day2007fair}. Then, by fixing this revenue,  it searches for another point in the current feasible polytope of core candidates (with one additional constraint for fixing the revenue) that has minimum $\ell_2$-distance to VCG. This search is done using convex quadratic programming. The algorithm iterates over this procedure until it finds a feasible core point. We use Matlab's large-scale LP solver based on the interior-point methods~\citep{zhang1998solving} for the linear programming part, and Matlab's large scale interior-point method for solving convex quadratic programming~\citep{mehrotra1992implementation,gould2004preprocessing} for the quadratic programming part of each iteration.
\item \underline{Vaidya's cutting plane method (\texttt{Vaidya Min Rev}; \citealp{vaidya1989new, vaidya1996new})}: one approach to find a minimum revenue core point in polynomial-time is to solve the LP for the minimum revenue directly, as we have access to a separation oracle for the core polytope. As a reminder, this oracle is essentially the allocation algorithm for the winner determination problem with truncated values, as in Definition \ref{def:win-det}. Vaidya's volumetric cutting plane method is a fast algorithm that makes efficient use of the separation oracle. If $n$ is the number of bidders, it has oracle complexity $O(n \log(n))$ and computational complexity $O(n^4)$. For details on different steps of this algorithm and its analysis see \citealp{bubeck2015convex}. We implemented Vaidya's algorithm by following the steps  in Section 2.3 of \citealp{bubeck2015convex} to be used in our numerical experiments.
\item \underline{Our bidder optimal core pricing algorithm (\texttt{Fast Core})}: we implemented our core pricing rule following the steps of Algorithm~\ref{alg:bidderopt}, which uses  Algorithm~\ref{alg:binary search} as a subroutine. 
\end{enumerate}

\begin{remark}
Note that we studied above auctions \emph{without} reserve prices. Importantly, tuned reserve prices are commonly used to boost revenue.  However, as reserves can be applied to all seven of these auctions, we choose not to include them in our experiment in order to focus on the impact of the auction pricing rule. 
\end{remark}
\begin{remark}
In fact, the computational complexity of Vaidya can even be improved further, and the recent breakthrough by \citealp{lee2015faster} shows that it can essentially (up to logarithmic factors) be brought down to $O(n^3)$. As we will elaborate in Section~\ref{sec:practical-considerations}, we already faced several practical complications and obstacles to implement Vaidya for our application (mostly due to the sensitivity of its performance to various parameters of the algorithm to make use of volumetric barrier). Hence, we left  implementing the algorithm of \citealp{lee2015faster} and adapting it for our application as a future direction. 
\end{remark}

\section{Practical considerations and limitations in our experimental study}
\label{sec:practical-issues}
Here we note a few important practical notes that apply to our numerical results of this section, together with potential interpretations or road-maps on how to deal with them based on various work in the literature. Further and deeper discussion on these points is beyond the scope of this work and we leave them as interesting open directions for future work.
\subsection{Bid collection vs. true valuations}
\label{sec:bid-collection}

During bid collection, the sponsored search platform was running its own native auction (which roughly speaking is a variant of the generalized second price auction with optimized reserve prices). These bids by no means are guaranteed to be truthful bids,  neither are collected through a controlled experiments in which bidders are aware of changing the auction to a core selecting auction or any other auction that we are simulating in this section.
Moreover, in online advertising, it is difficult to evaluate an advertiser's true valuation of a click from the submitted bids,  because of several fundamental reasons: (1) Importantly, advertisers usually run sophisticated learning algorithms to bid based on their past experiences with the platform,  which leads to complicated bid shading mechanics, (2)  Indeed, the advertiser itself might not know the true valuations for the clicks received from the search engine, as the quality of clicks differs and is dependent on the user, publisher, and other contexts, (3) Advertisers might not be utility maximizers, for example they might be maximizing clicks given a budget, or might be valuing the conversion (when a defined transaction such as a sale or subscribing happens from the user after the click) or just a visit to their page, and hence their behavior diverge from the classic quasi-linear rational models in microeconomics, and finally (4) Even if the appropriate metrics were known, the platform typically does not get to observe them accurately. Understanding and modeling true advertisers' behavior is still an important open problem for ad auction research and industry, and is beyond the scope of this paper~\citep{edelman2007strategic,xu2013predicting,sayedi2018real}. 

Nevertheless, in the absence of knowing the exact behavior of advertisers and the possibility of running a controlled exclusive experiment, the bidding numbers are our best estimates for the advertisers' true values for clicks. Therefore, when using bidding data to evaluate our algorithms, we implicitly assumed truthful bidding and took each advertiser's bid as a proxy for their value.  We made this choice in part due to necessity, and in part because it is a common practice in search advertising auction industry. As a minor note, there are also strong empirical evidences for modeling advertisers as ROI (Return of Investment) constrained value maximizers (e.g., see \citealp{aggarwal2006truthful,wilkens2017gsp}), under which  GSP is provably a truthful auction~\citep{wilkens2017gsp}. Yet, these results heavily rely on myopic rational behavior for advertisers, which might not be the case in practice.

\subsection{ Non-truthfulness in core auctions and short/long-term incentive issues }
\label{sec:incentive-core}
One important property that core selecting auctions lack is truthfulness. So, even though in our numerical experiments the reported bids are acceptable proxies for the true valuations, it is critical to understand the behavior of bidders, both in short-term and long-term, and how they respond to the non-truthfulness of the auction. This response can have implications on the revenue of the auction. For example, when bidders bid strategically or run a learning algorithm to respond, the platform's revenue could be hurt. We propose the following interpretations and methodologies to study this phenomena. Digging deeper in some of these methodologies is beyond the scope of this work and we leave as future research directions:
\begin{itemize}
\item Running a bidder optimal core selecting auction imposes a natural full information Nash equilibrium, where bidders truncate their values by the utility they get from the bidder optimal core pricing based on the the core with respect to the \emph{true} valuations~\citep{day2008core}. Moreover, the revenue of the seller at this equilibrium is equal the revenue of the bidder optimal core pricing with respect to the true valuations. As we are running bidder optimal core selecting auctions and advertises have not responded to this auction (so, bids are proxies of true valuations), one can therefore interpret our experimental results as evaluating the \emph{short-term} impact of employing core pricing, modulo the assumption that advertisers are playing the aforementioned full information Nash equilibrium above.
\item Assuming that advertisers play the full information Nash equilibrium can be problematic in sponsored search auctions, as these markets are highly uncertain and information about competitors is incomplete. In such an environment, an advertiser might run a learning algorithm or some other dynamic responding mechanism to bid. However, there are evidences that bidder optimal core auctions ``guide the advertisers" who run natural dynamics, in a way that advertisers converge fast to this full information Nash equilibrium, without actually needing complete information. To see this, consider the following two-step process: in the first step, advertisers bid truthfully (as they have no information about other competitors) and each winner obtains some utility. Now, as the auction is a bidder optimal core auction, these utilities will reveal the full information Nash equilibrium mentioned earlier. In fact, each bidder only needs to play a truncated strategy that shades the bid of each package by the utility obtained in the first step (which is a somehow natural bid shading algorithm), and in this way the full information Nash equilibrium will be played in the second step. One can also think of other generalizations of the mentioned two-step dynamic, where the bidders only shade their bid by truncating based on a fraction of the utility of the previous round. We conjecture this simple dynamics, as well as other dynamics such as iterative best response, converge to the full information Nash equilibrium in the sale of ad space problem, and we leave investigating further as an open question (interestingly, there are examples showing that for a general combinatorial auction this dynamics does not converge, and at the same time there is a simple proof showing that for simple settings such as single-minded combinatorial auctions this dynamics converges to the full information Nash equilibrium).
\item Another possible approach to interpret the the long-term effects of the incentive issues of core auctions is to focus on an incomplete information equilibrium  concept such as  Bayes Nash Equilibrium (BNE). While characterizing closed-form equilibria is intractable (and probably impossible), one can use the recent progress on computational methods to approximate the BNE of core selecting auctions~\citep{bosshard2017computing,lubin2009quantifying,lubin2015new,bunz2015faster,bunz2018designingEC,bunz2018designing}, in order to obtain an educated guess ball-park characterization of how far we might expect the revenue to be from the simulation results in practice (basically after platform uses this auction and advertisers respond to it in a way that they approximately land in the aforementioned BNE). \citealp{bunz2018designing,bunz2018designingEC} used the computational search approach to help with better understanding the quadratic core payment rule, as well as designing an automatic search for ``better" core payment rules. This paper cannot be used in a blackbox fashion to assess our core pricing rule (as one needs to run the computational search approach for our core pricing rule. Also they have considered different domains than ours and revenue is very domain dependent). Yet, it can be used to obtain very rough estimates of how much revenue reduction we might face at the (approximate) BNE. By looking at the revenue column of Tables~1-7 in \citealp{bunz2018designing,bunz2018designingEC}, it seems the revenue ratio of quadratic payment rule over VCG is a number between $0.88$-$1.44$ among 7 different domains (which is averaged at $11.27\%$ improvement). Our numerical experiments suggest an improvement around $15\%$ (see Table~\ref{tab:revenue}). One might try to repeat the same approach, but tailored to our core selection algorithm and our domain, and refine the $26.5\%$ improvement over VCG that our experimental results suggest for the revenue of our core pricing algorithm. We leave this refined experimental study as an interesting future direction. 


\end{itemize}

\subsection{Using approximate winner determination for faster auctions}
\label{rem:second}
We described a dynamic program for generating a welfare-optimal slate of ads, and showed that it was feasible within the time constraints imposed by a production advertising platform for a realistic number of ads and lines.  However, if the number of advertisers or lines increases, and/or there is a change on the assumptions that can be made on buyer utilities, it may turn out that the dynamic programming solution is not always feasible in practice.  If not, we note that one can replace the optimal allocation with an approximate Maximal-In-Range allocation (see \citealp{dobzinski2007mechanisms}) and get the same result as in this paper, but for the core polytope restricted to that range.  In other words, if one can compute the optimal allocation within a restricted range of allocations (e.g., those that allocate at most $4$ ads, or that only show ads of the same size, etc.), then our algorithm -- using such a restricted welfare maximization oracle -- will generate core payments for the auction that restricts outcomes to lie in that range.

\subsection{Comments on implementing cutting plane methods}
\label{sec:practical-considerations}
We observed the following issues while implementing Vaidya's algorithm~\citep{vaidya1989new,vaidya1996new} for the sale of ad space application (we use notations used in Section 2.3 of \citealp{bubeck2015convex}), which we believe makes this algorithm (and other cutting plane methods with similar operations) impractical for our purpose:
\begin{enumerate}
\item Vaidya’s algorithm hugely depends on the precision handling, where precision controls ``how much the final point is close to a minimum revenue core point". In fact, it is very susceptible to overflow as it needs to compute the logarithm function for different ``slack" terms, and overflow easily happens when these slack terms approach to zero. Hence, in practice, the precision has to be very lax in order for the algorithm to converge.
\item Vaidya’s algorithm convergence hugely depends on different parametric choices (e.g., when to break the algorithm based on the volume of the search region, choice of parameter $\beta$ in each iteration, etc.). So, even with a lax precision it takes a long time to converge in some cases.
\end{enumerate}

In implementing Vaidya’s algorithm, we note that as we increase the precision of the solution slightly, its running time degrades significantly. Our experiments show that the running time scales by a factor comparable to $O(\tfrac{1}{\epsilon})$ in order to obtain an small precision of $\epsilon$. This is in contrast to our algorithm where running time only scales by $O(\log(\tfrac{1}{\epsilon}))$ to get the precision $\epsilon$ (and our experiments validate this fact as well). Another important factor in practice is the ignored constants of the big O notation in the running time. We again argue that comparing to our algorithm Vaidya’s algorithm has much larger constants in its running time, as suggested by our experiments and carefully implementing the algorithm following the recipe in Section 2.3 of \citealp{bubeck2015convex}). As a result, even obtaining a point with a rather large precision of $\epsilon=0.1$ requires a long processing time.

To better understand the reasons behind the precision sensitivity, which are also true for other barrier-based cutting plane methods such as \citealp{lee2015faster},  consider the following. The Vaidya’s algorithm has three main parameters that govern the precision of its final solution (and therefore its running time). The first parameter is the stopping threshold for the volume of the search polytope (measured indirectly using the volumetric barrier), below which we stop reducing the feasible region and terminate. The second parameter governs how small the leverage score of a hyperplane associated to a face of the search polytope should be to be removed (see Section 2.3.2 of \citealp{bubeck2015convex}, line (1) in the description of the algorithm for more details). Finally the last parameter is the standard floating point precision which is used for floating point operations in practice. Our experiments (and further playing around with Vaidya's algorithm to optimize its running time) suggest that Vaidya’s running time is very sensitive to the choice of the first and second parameters above, mostly due to the logarithm function used in the volumetric barrier.

\section{Future directions and open problems}
\label{apx:openproblems}
Here is a list of a few concrete theoretical and practical future directions:
\begin{enumerate}
\item Our fast core pricing algorithm (Algorithm~\ref{alg:bidderopt}) does not target any \emph{global} objective function, which is in contrast to other traditional methods such as \cite{day2007fair} (minimizing $\ell_1$ distance from VCG) and \cite{day2012quadratic} (minimizing $\ell_2$ distance from VCG). Can one use the way we explore the core polytope efficiently to optimize a tractable (maybe convex) objective function in a fast fashion?
\item Give the practical considerations in Section~\ref{sec:practical-considerations}, how can one adapt the state of the art cutting plane methods such as \cite{lee2015faster} to work for the sale of ad space problem? Our simulations suggest these algorithms are very sensitive to different precision parameters. Designing a robust version of them sounds like a roadmap to approach this challenge.
\item In sponsored search, advertisers either use sophisticated online learning algorithms, or are relying on automated bidding softwares provided by the platform~\citep{aggarwal2019autobidding}. What can we say about the way they respond to a core selecting auction, in comparison to GSP auction and other auction formats that are common in sponsored search?
\item Designing auctions for video ads is another combinatorial setting that can potentially benefit from a combinatorial auction such as core selecting. An interesting open problem is to use the techniques in our paper to tackle this multi-billion dollar industry problem.
\end{enumerate}
\section{More experimental results for various line counts}
\label{appendix:experiments}
We report the revenue, fairness, running time and query complexity of different algorithms in our experiments for line counts 25, 30, 35, and 45. For revenue results, see Figures~\ref{fig:revenue-25}, \ref{fig:revenue-30}, \ref{fig:revenue-35}, and \ref{fig:revenue-45}. For running time results, see Figures~\ref{fig:runtime-25}, \ref{fig:runtime-30}, \ref{fig:runtime-35}, and \ref{fig:runtime-45}. For query complexity results, see Figures~\ref{fig:query-25}, \ref{fig:query-30}, \ref{fig:query-35}, and \ref{fig:query-45}. Finally, for fairness results, see Figures~\ref{fig:fairness-25}, \ref{fig:fairness-30}, \ref{fig:fairness-35}, and \ref{fig:fairness-45}.

\begin{figure}[h]
\hspace*{-2cm}
\includegraphics[width=6.3in]{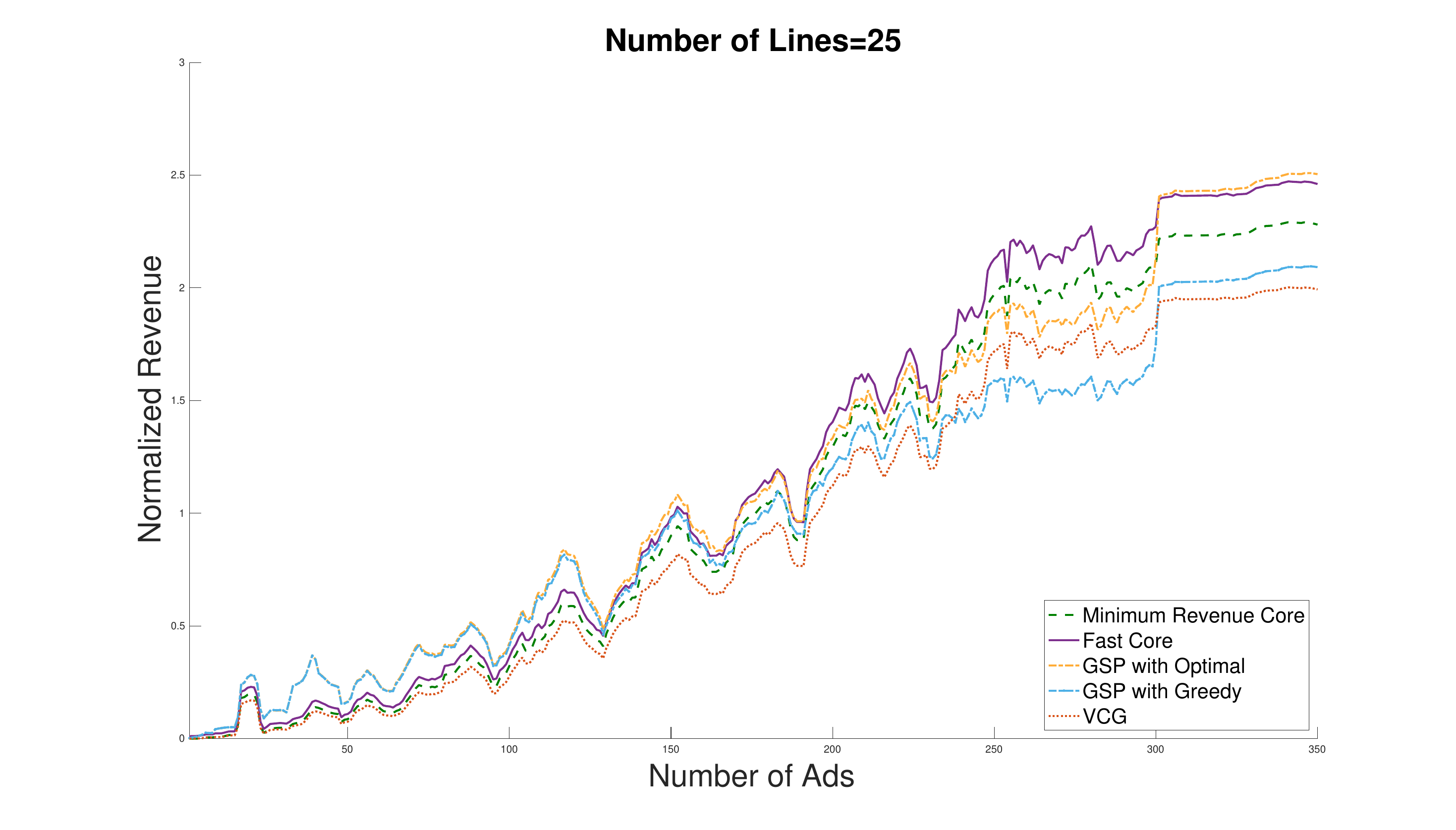}
\centering
\caption{Normalized revenues versus total number of ads for line count=25.\label{fig:revenue-25}}
\end{figure}

\begin{figure}[h]
\hspace*{-2cm}
\includegraphics[width=6.3in]{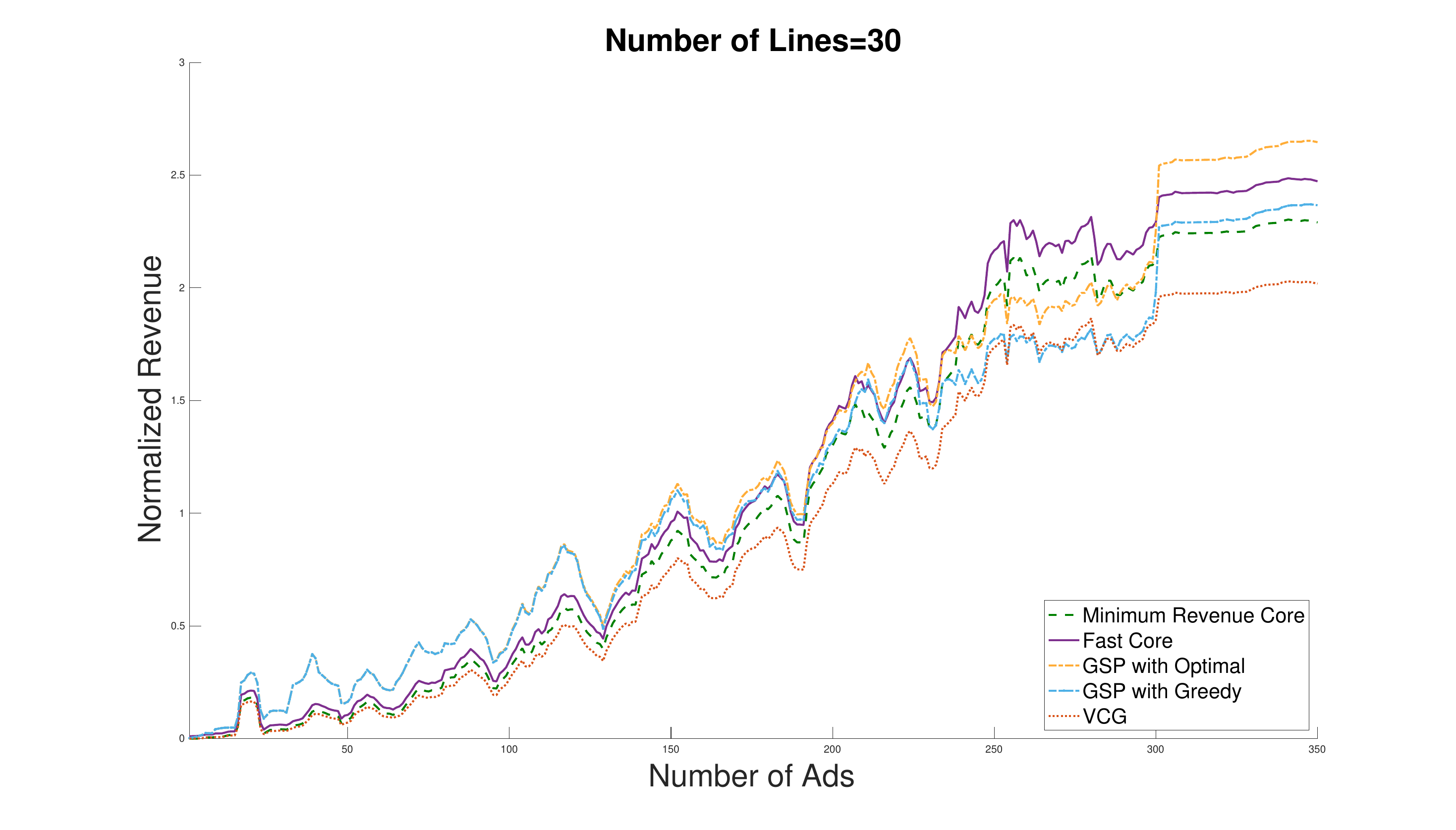}
\centering
\caption{Normalized revenues versus total number of ads for line count=30.\label{fig:revenue-30}}
\end{figure}

\begin{figure}[h]
\hspace*{-2cm}
\includegraphics[width=6.3in]{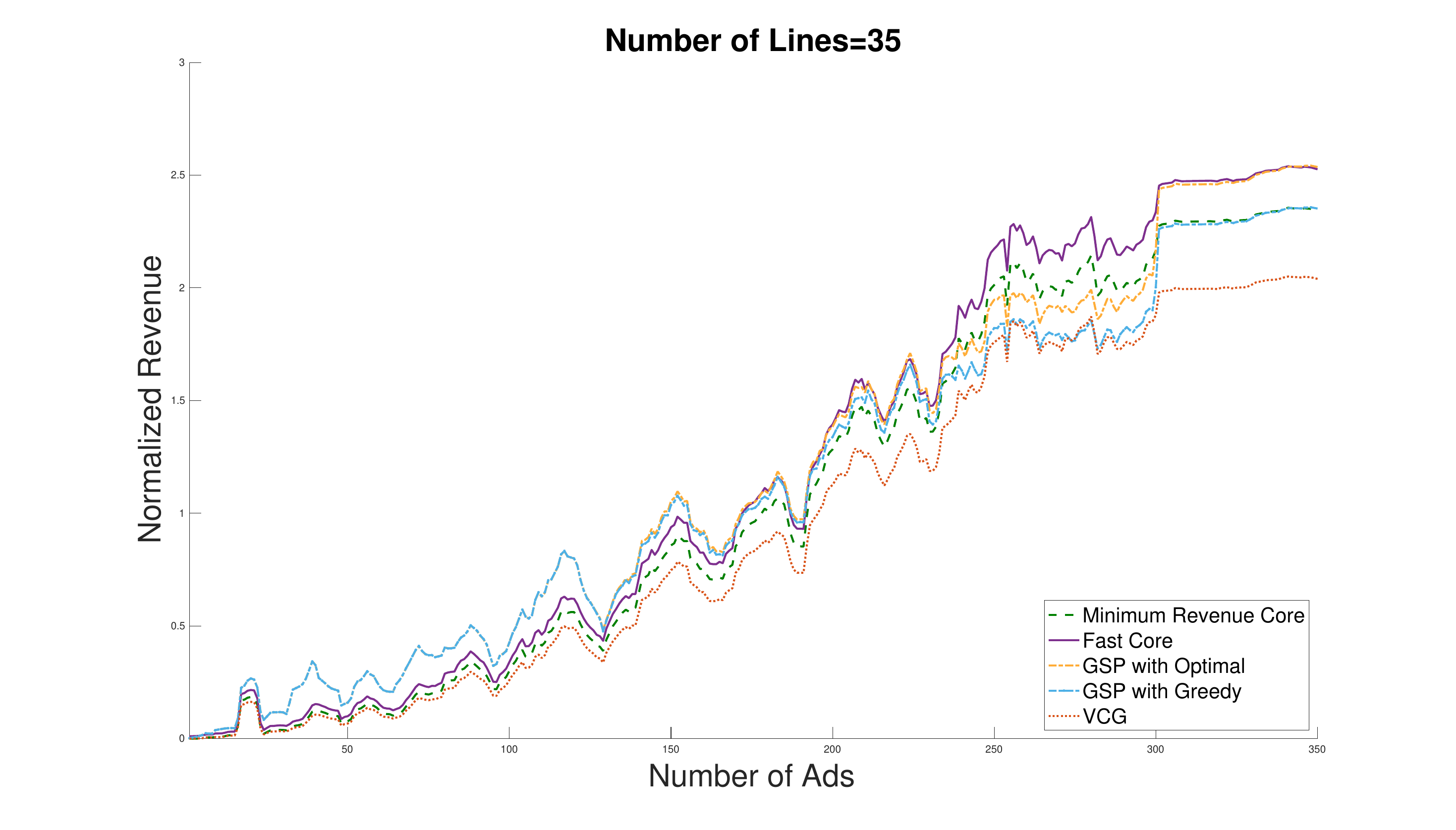}
\centering
\caption{Normalized revenues versus total number of ads for line count=35.\label{fig:revenue-35}}
\end{figure}

\begin{figure}[h]
\hspace*{-2cm}
\includegraphics[width=6.3in]{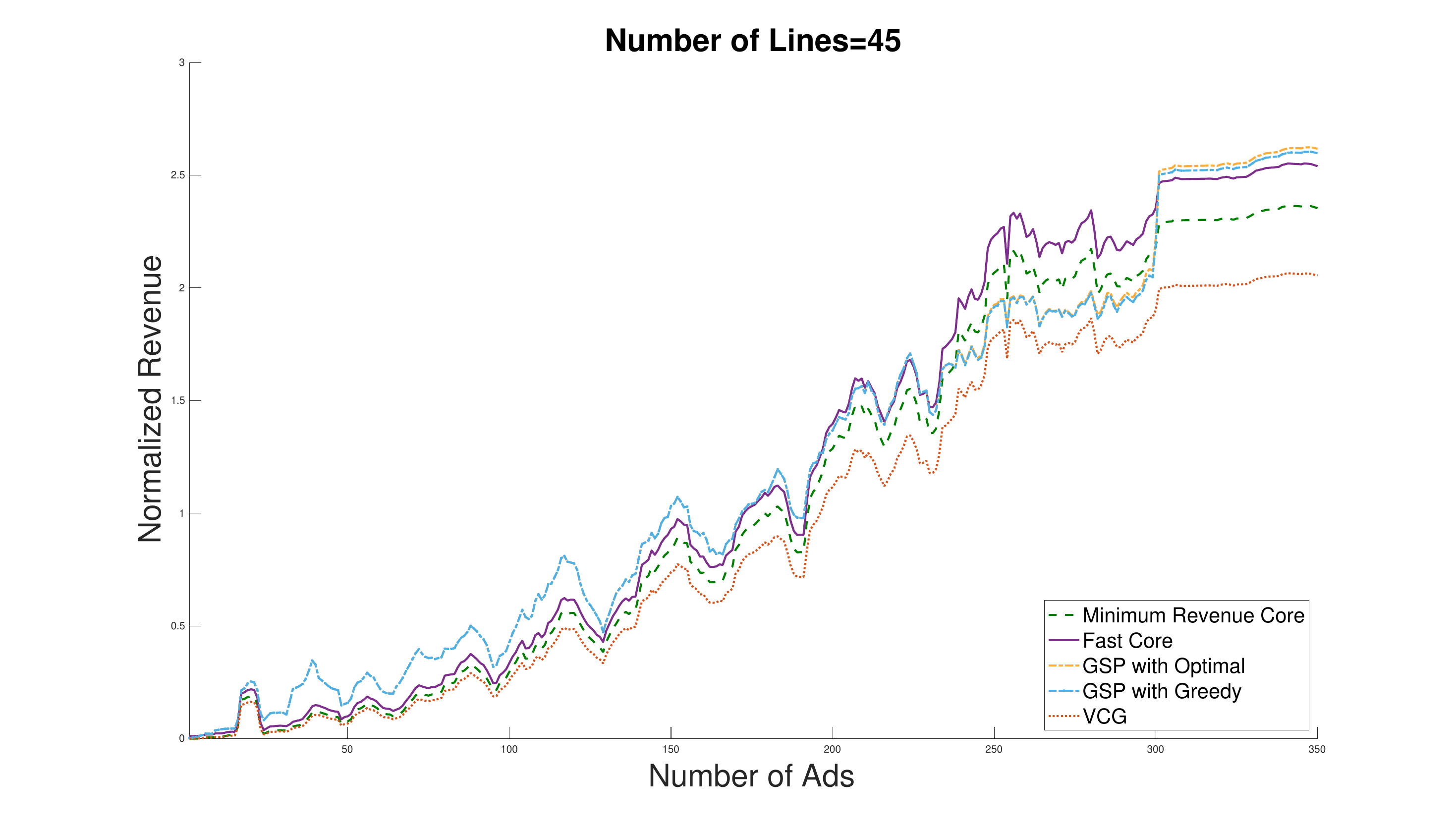}
\centering
\caption{Normalized revenues versus total number of ads for line count=45.\label{fig:revenue-45}}
\end{figure}

\begin{figure}[h]
\hspace*{-2cm}
\includegraphics[width=6.3in]{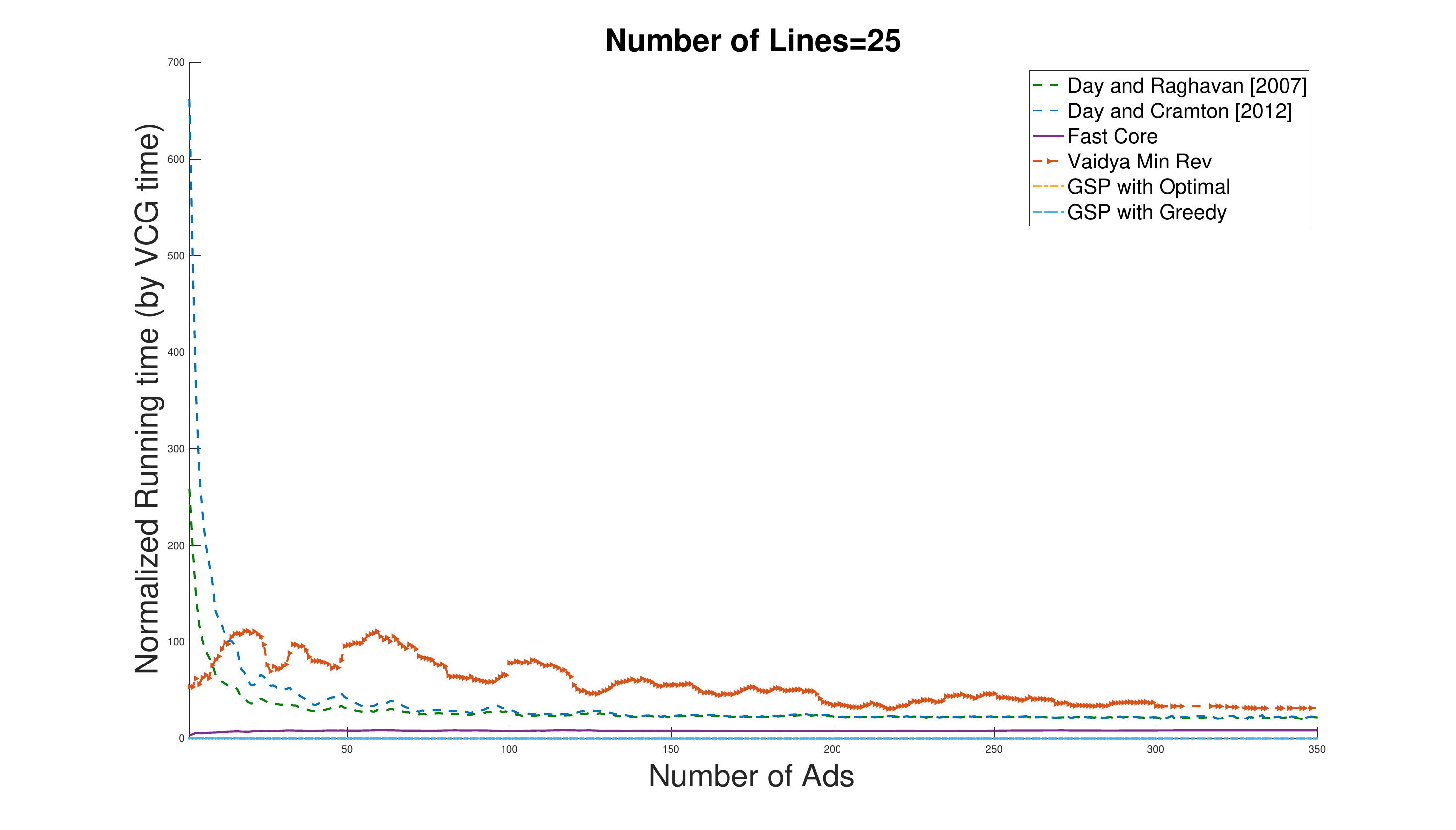}
\centering
\caption{Normalized running times versus total number of ads for line count=25.\label{fig:runtime-25}}
\end{figure}
\begin{figure}[h]
\hspace*{-2cm}
\includegraphics[width=6.3in]{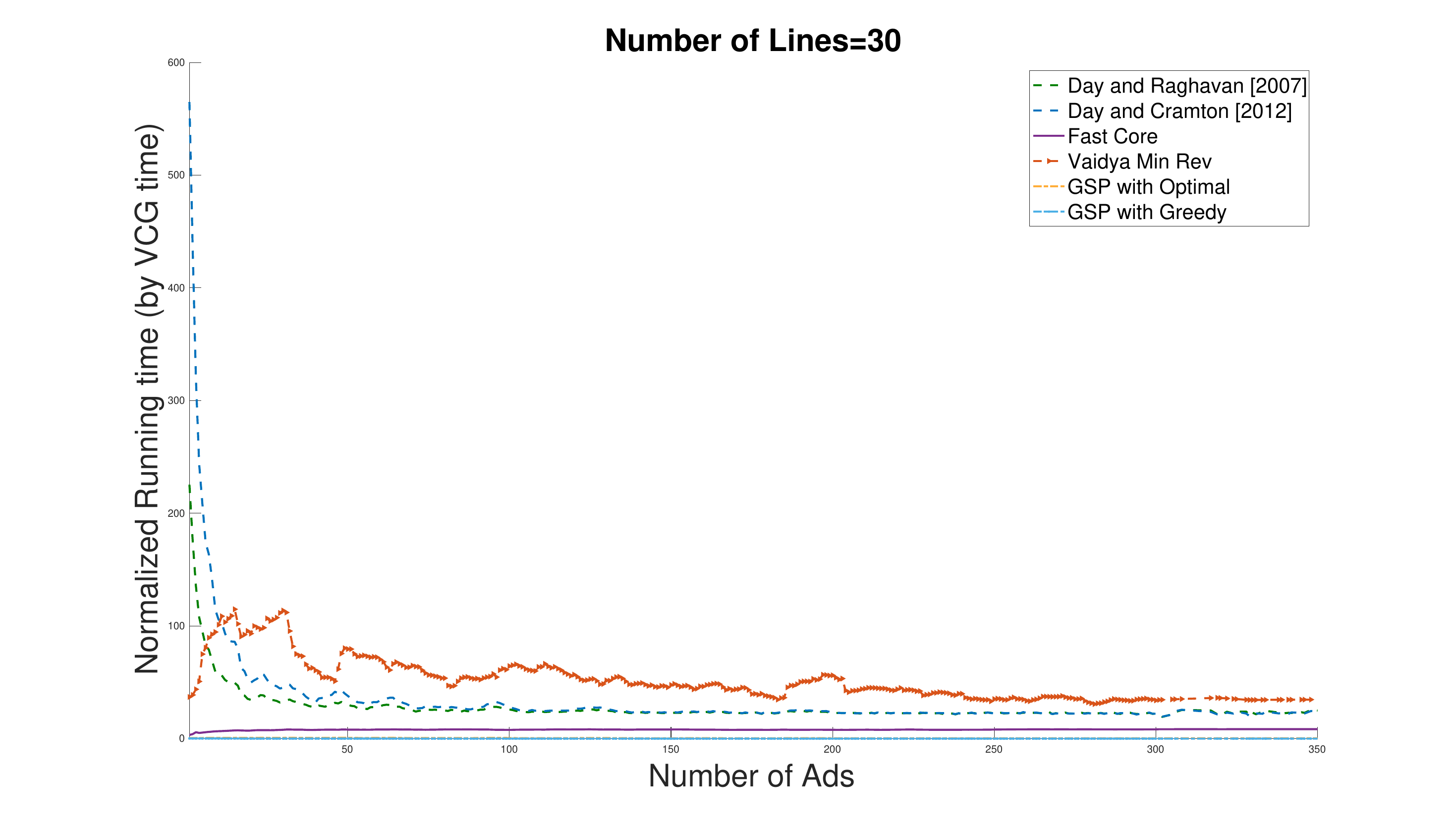}
\centering
\caption{Normalized running times versus total number of ads for line count=30.\label{fig:runtime-30}}
\end{figure}

\begin{figure}[h]
\hspace*{-2cm}
\includegraphics[width=6.3in]{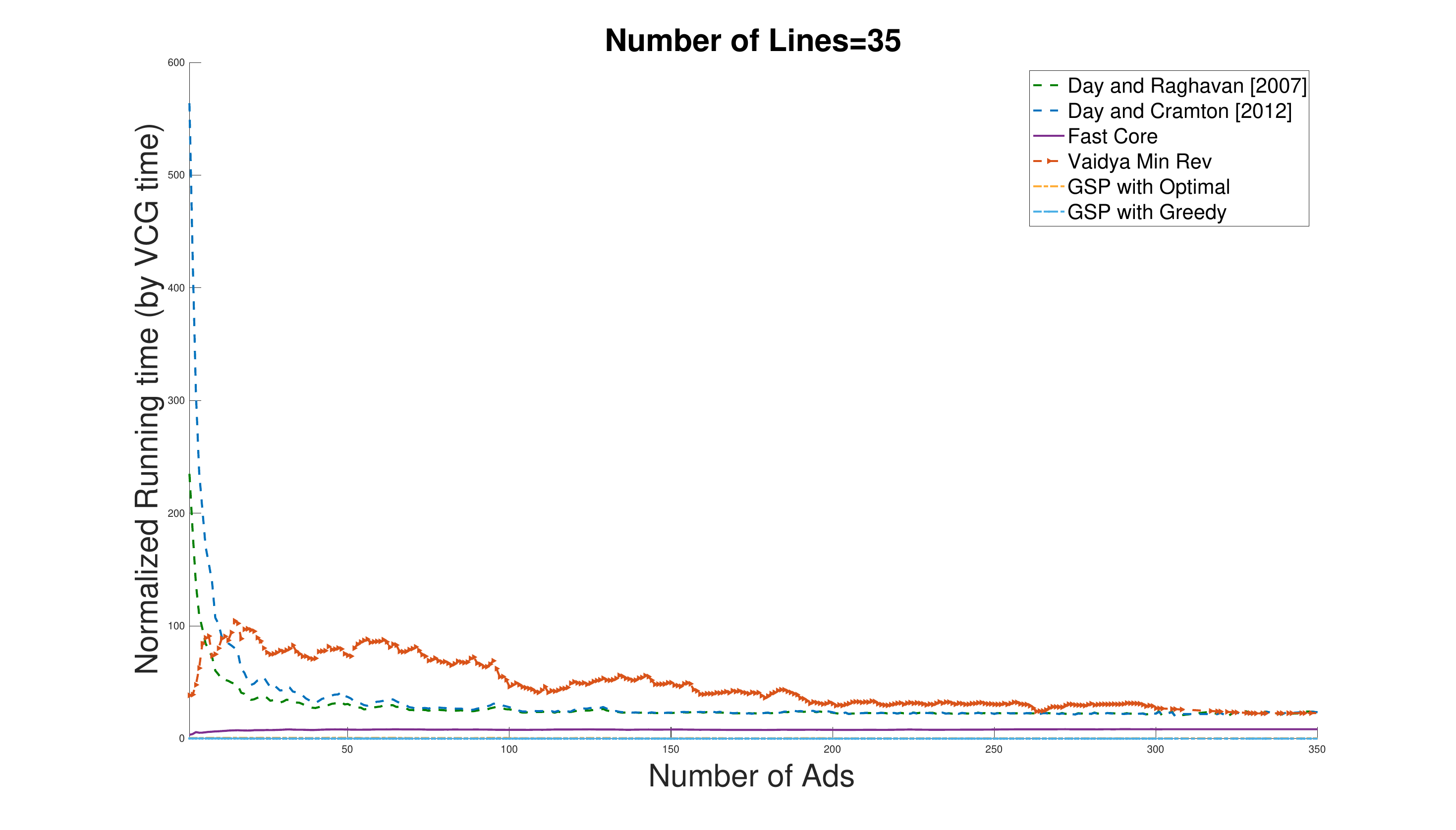}
\centering
\caption{Normalized running times versus total number of ads for line count=35.\label{fig:runtime-35}}
\end{figure}
\begin{figure}[h]
\hspace*{-2cm}
\includegraphics[width=6.3in]{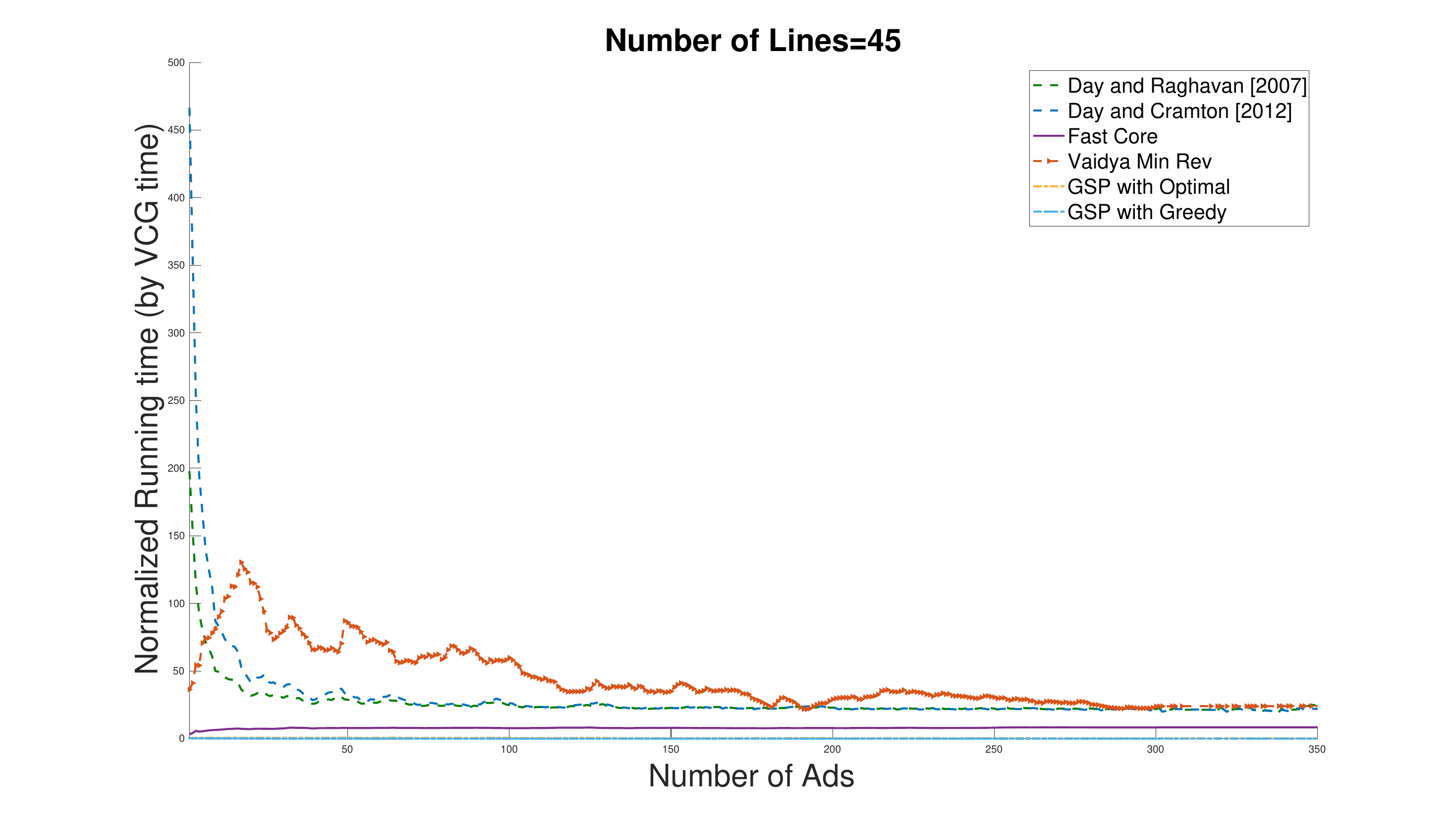}
\centering
\caption{Normalized running times versus total number of ads for line count=45.\label{fig:runtime-45}}
\end{figure}

\begin{figure}[h]
\hspace*{-2cm}
\includegraphics[width=6.3in]{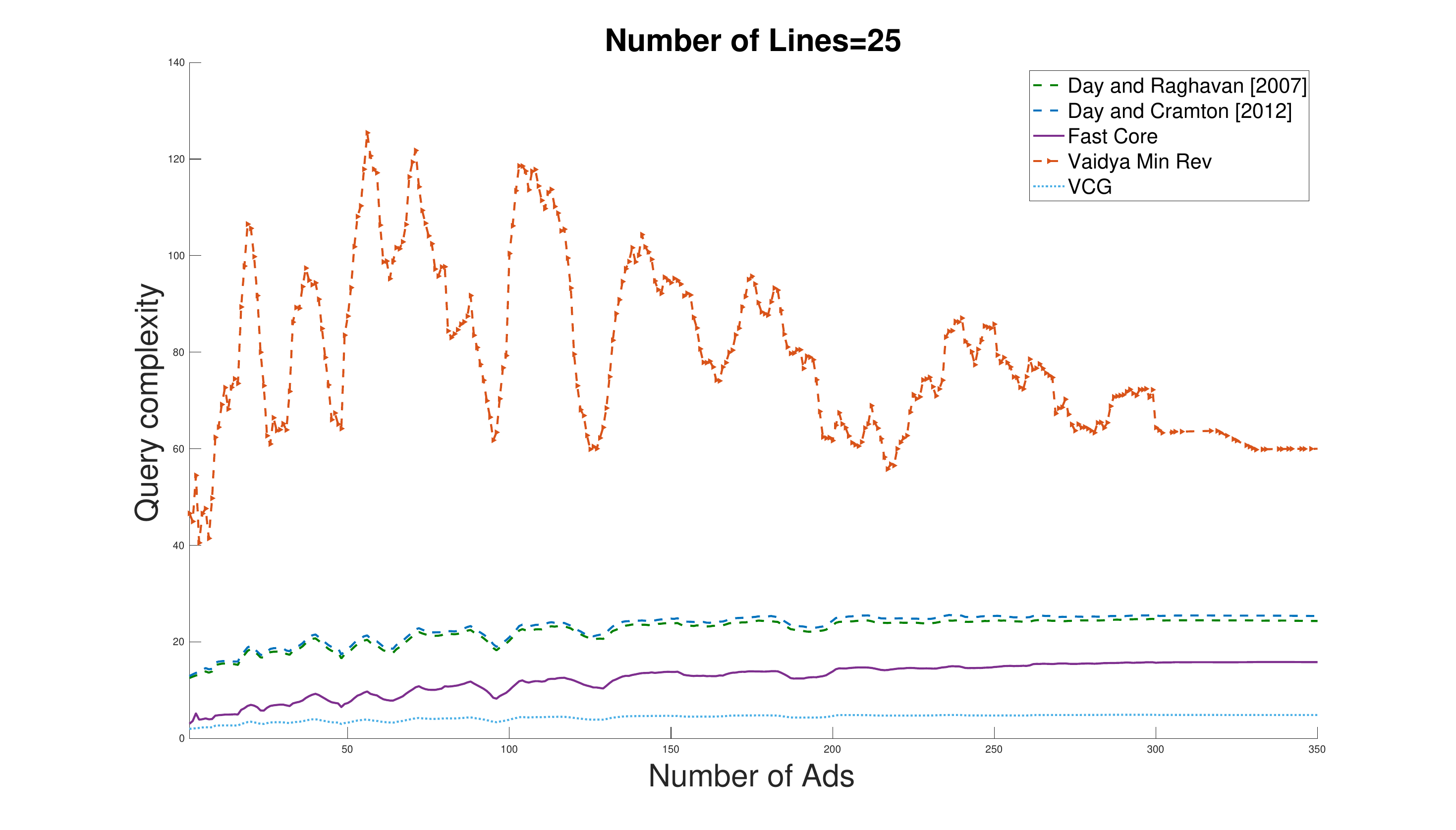}
\centering
\caption{Query complexity versus total number of ads for line count=25.\label{fig:query-25}}
\end{figure}

\begin{figure}[h]
\hspace*{-2cm}
\includegraphics[width=6.3in]{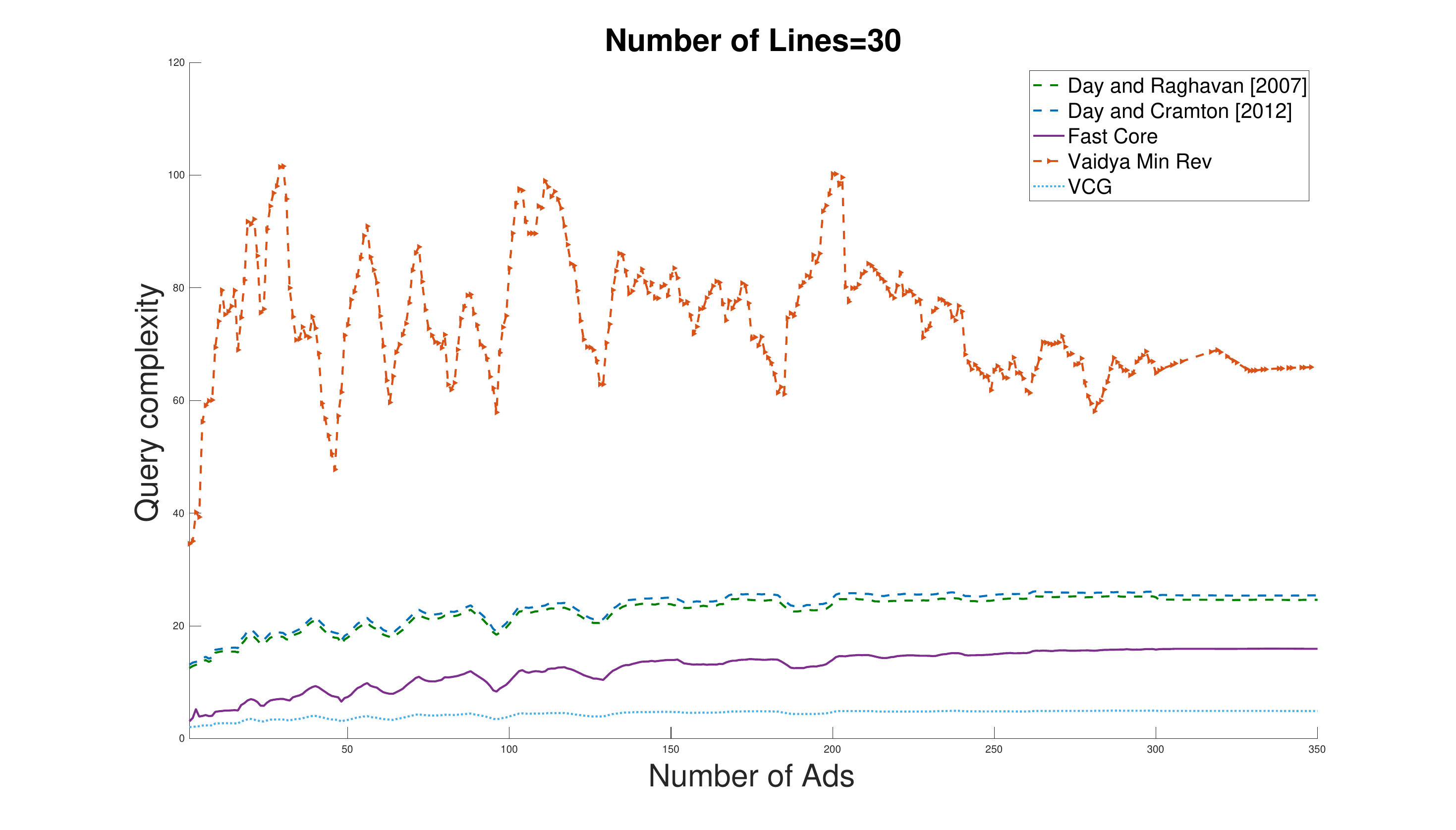}
\centering
\caption{Query complexity versus total number of ads for line count=30.\label{fig:query-30}}
\end{figure}

\begin{figure}[h]
\hspace*{-2cm}
\includegraphics[width=6.3in]{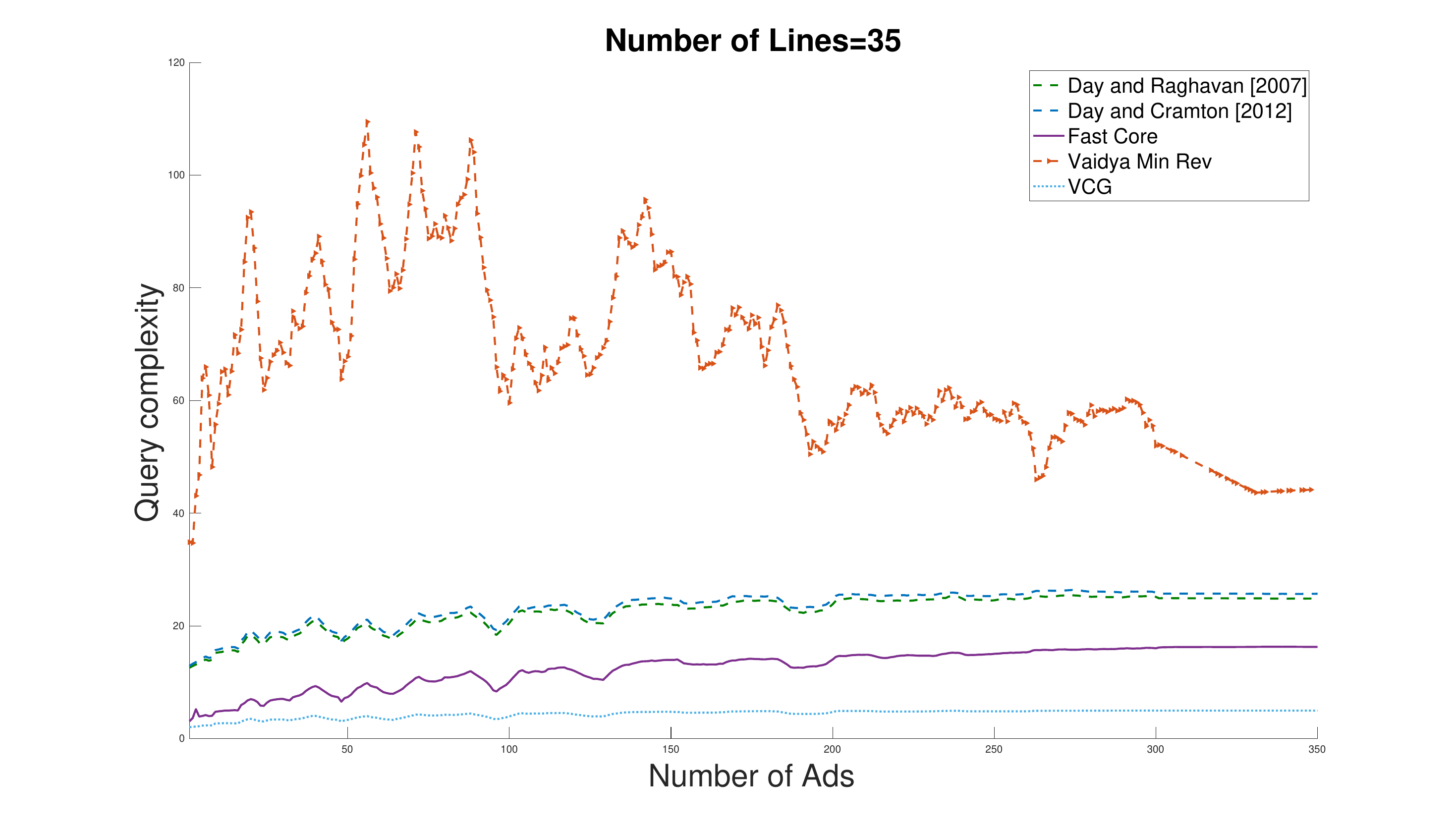}
\centering
\caption{Query complexity versus total number of ads for line count=35.\label{fig:query-35}}
\end{figure}

\begin{figure}[h]
\hspace*{-2cm}
\includegraphics[width=6.3in]{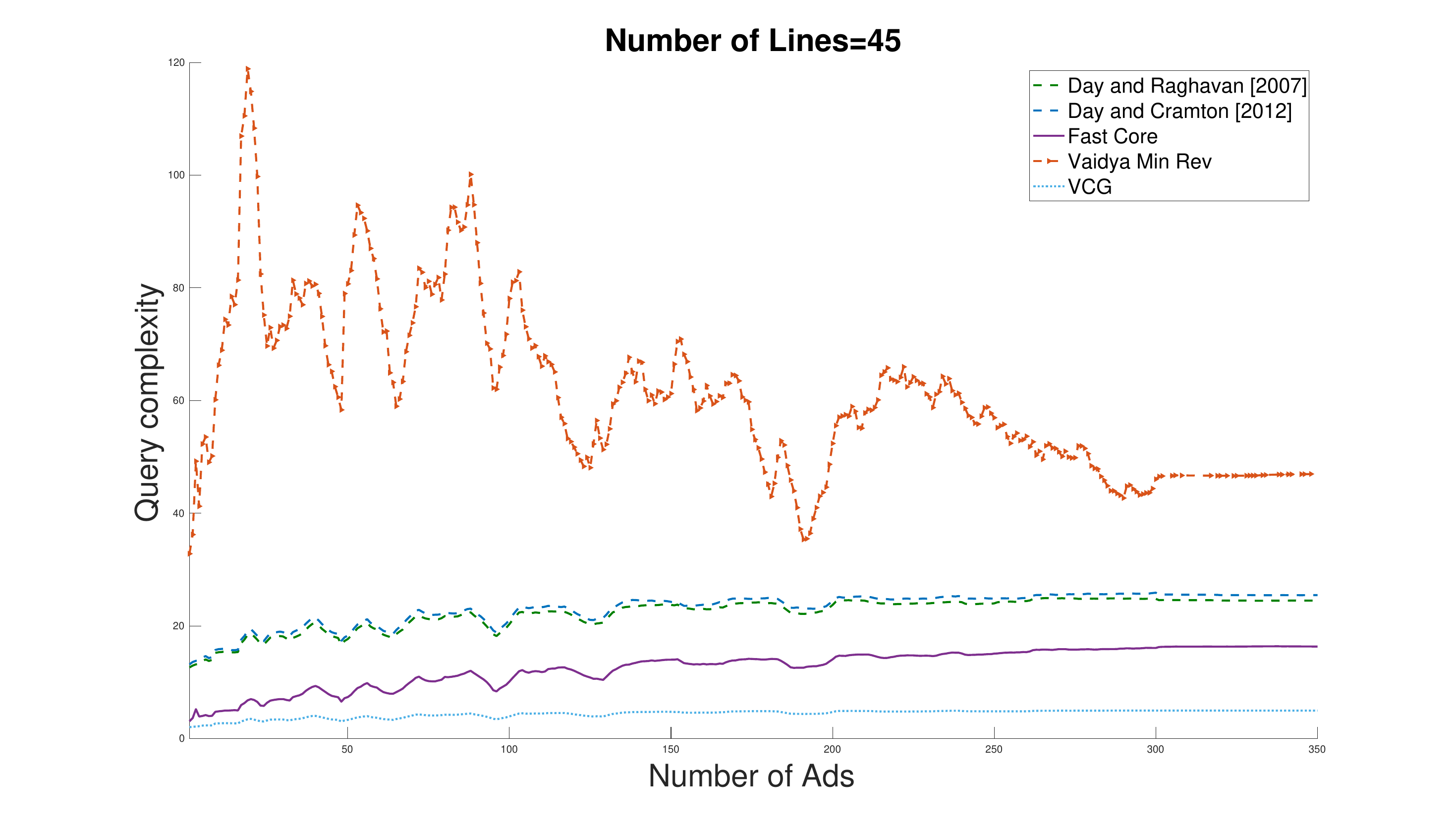}
\centering
\caption{Query complexity versus total number of ads for line count=45.\label{fig:query-45}}
\end{figure}

\begin{figure}[h]
\hspace*{-2cm}
\includegraphics[width=6.3in]{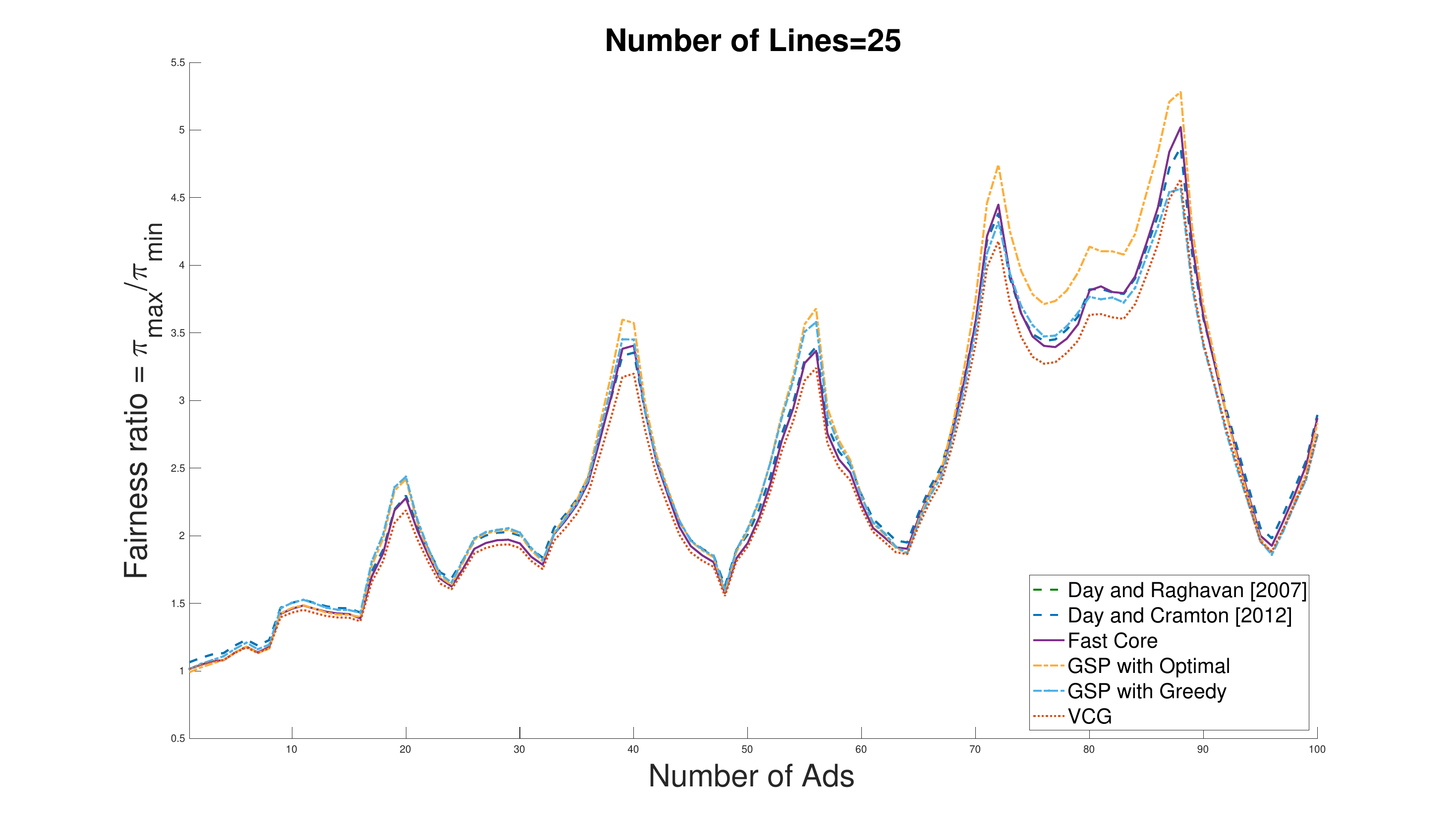}
\centering
\caption{Fairness ratio versus total number of ads for line count=25.\label{fig:fairness-25}}
\end{figure}

\begin{figure}[h]
\hspace*{-2cm}
\includegraphics[width=6.3in]{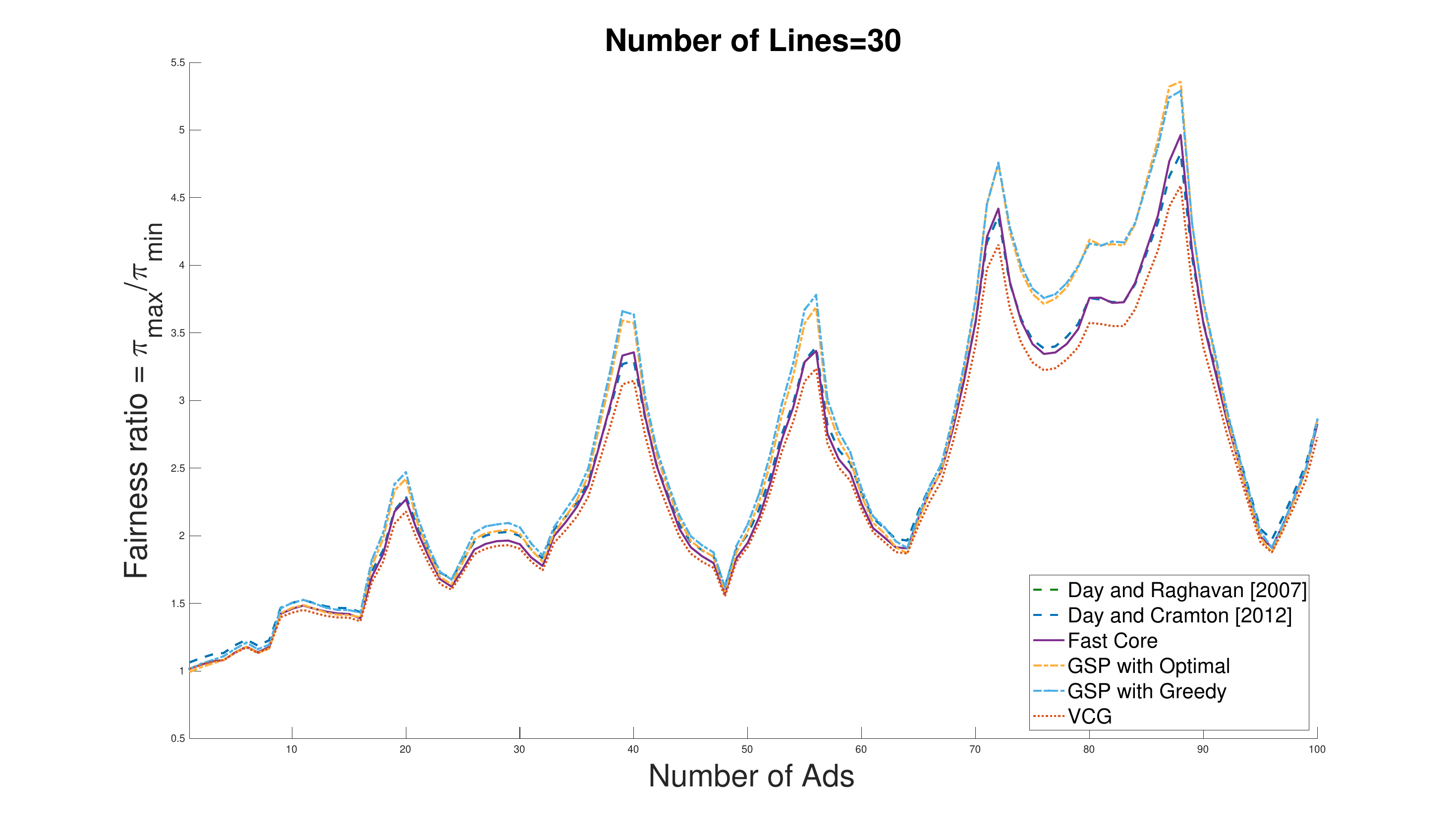}
\centering
\caption{Fairness ratio versus total number of ads for line count=30.\label{fig:fairness-30}}
\end{figure}
\begin{figure}[h]
\hspace*{-2cm}
\includegraphics[width=6.3in]{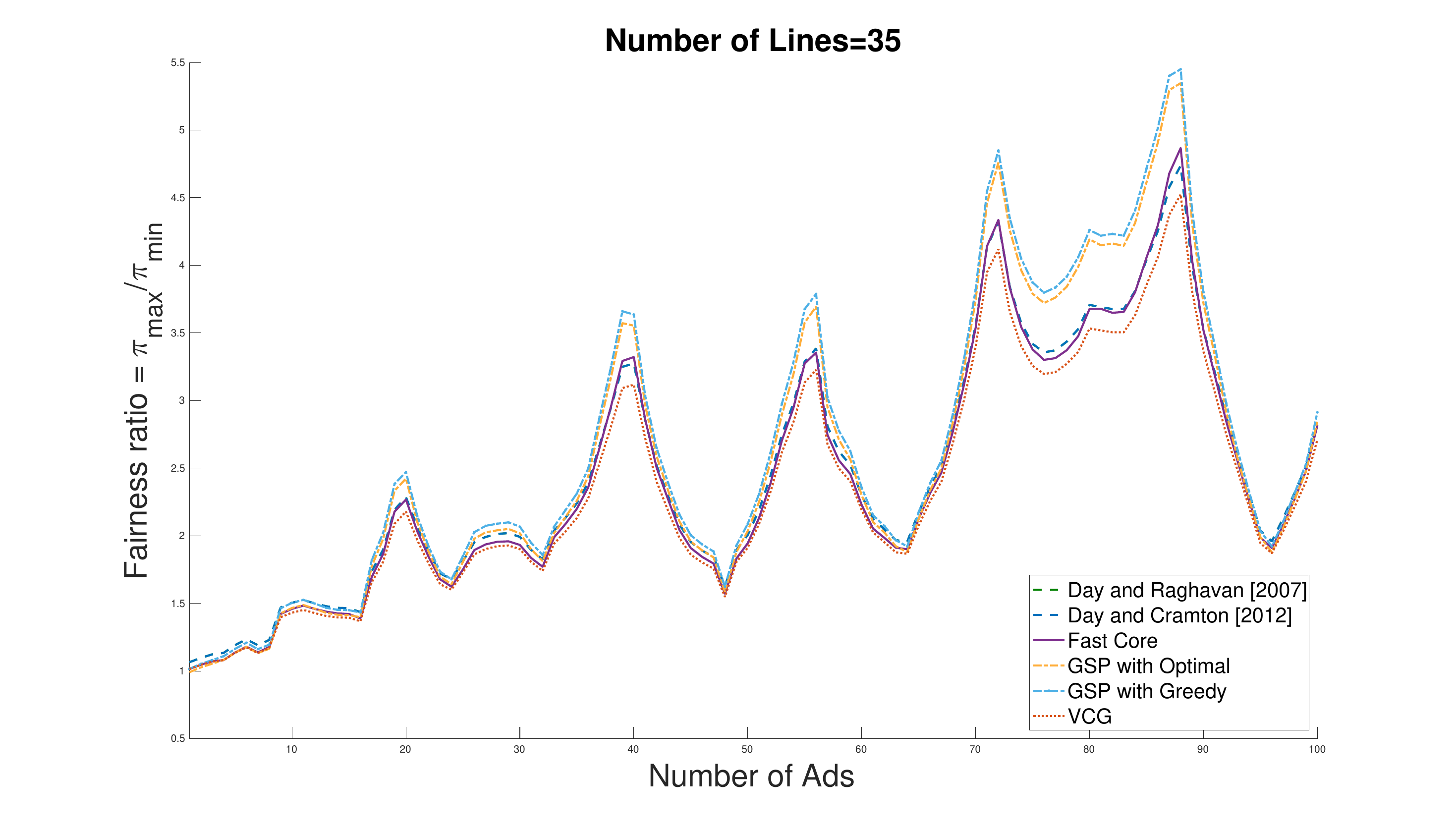}
\centering
\caption{Fairness ratio versus total number of ads for line count=35.\label{fig:fairness-35}}
\end{figure}
\begin{figure}[h]
\hspace*{-2cm}
\includegraphics[width=6.3in]{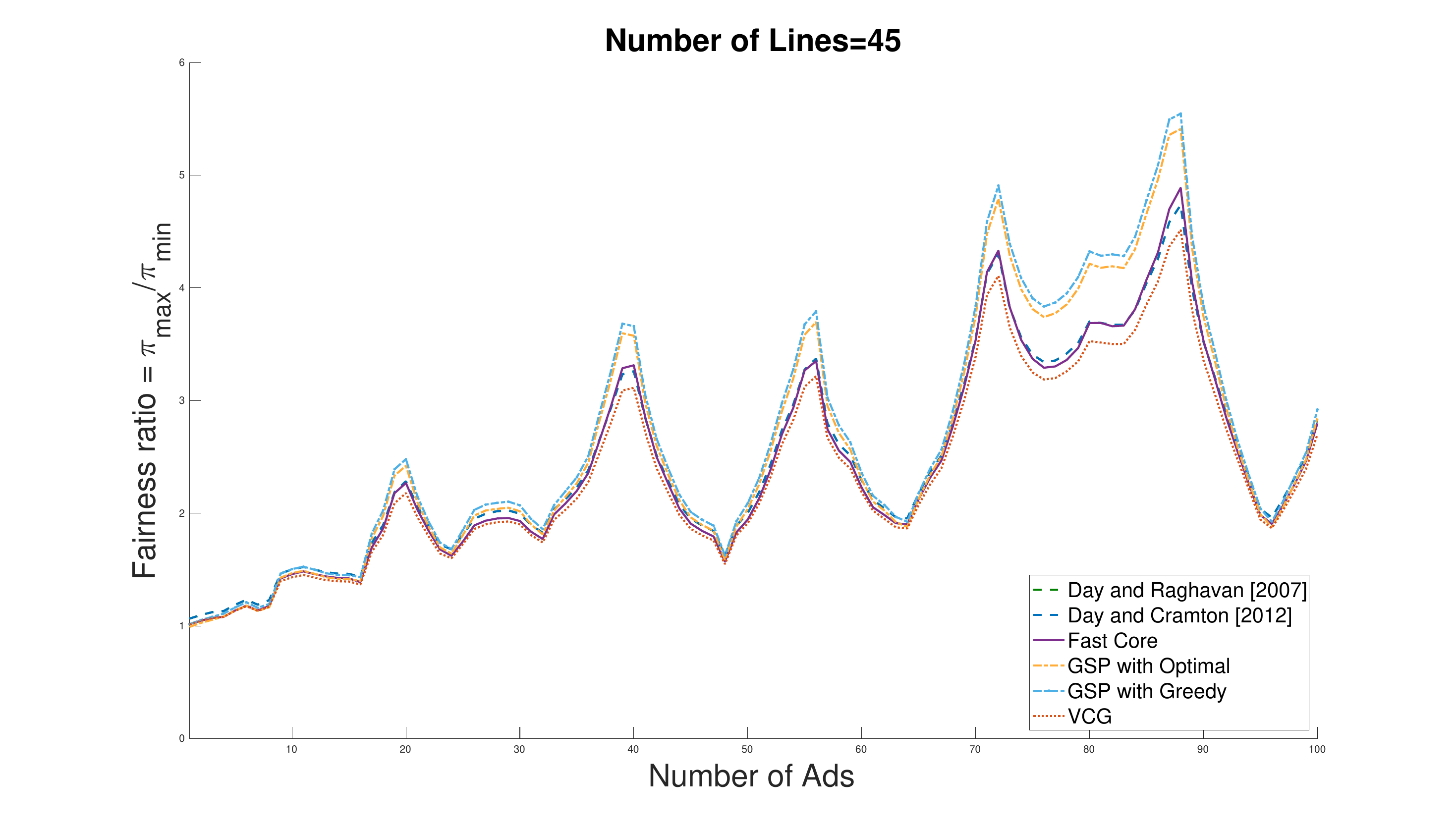}
\centering
\caption{Fairness ratio versus total number of ads for line count=45.\label{fig:fairness-45}}
\end{figure}

 \end{APPENDICES}
\end{document}